\documentclass[a4paper,10pt]{article}
\usepackage[cm]{fullpage}
\usepackage{amsmath,amsfonts,amssymb}
\usepackage{amsthm}
\usepackage{trimclip,adjustbox}
\newtheorem{theorem}{Theorem}[section]
\newtheorem{lemma}[theorem]{Lemma}
\theoremstyle{definition}
\newtheorem{definition}[theorem]{Definition}
\newtheorem{example}[theorem]{Example}
\newtheorem{proposition}[theorem]{Proposition}
\newtheorem{corollary}[theorem]{Corollary}

\usepackage[T1]{fontenc}
\usepackage[utf8]{inputenc}
\usepackage{hyperref}
\usepackage{circledsteps}

\hypersetup{
  colorlinks=true,
  linkcolor=blue,
  anchorcolor=blue,
  citecolor=blue,
  urlcolor=blue,
  filecolor=blue,
  breaklinks=true
}
\usepackage{xargs}
\usepackage{array}
\usepackage{stmaryrd}
\usepackage{scalerel}
\usepackage{proof}
\usepackage{needspace}
\usepackage{tikz}
\usepackage{tikz-cd}
\usepackage{float}
\floatstyle{boxed}
\restylefloat{table}
\restylefloat{figure}
\usepackage{physics}
\usepackage{tcolorbox}
\usepackage{xfrac}
\usepackage{mathrsfs} 
\usepackage{paralist}
\usepackage{sbnf} 
\renewcommand{\bnfcmt}[1]{\textbf{(#1)}}
\def\bnfmid{\text{\,\Large$\mid$\,}}
\usepackage{tarrow} 
\usepackage{infer}
\renewcommand*{\InferLabelStyle}[1]{\textcolor{violet}{\ensuremath{[\mathtt{#1}]}}}
\usepackage{cleveref}

\usetikzlibrary{positioning,decorations.pathreplacing}
\makeatletter
\pgfdeclareshape{xor}
{
  \inheritsavedanchors[from=circle]
  \inheritanchorborder[from=circle]
  \inheritanchor[from=circle]{north}
  \inheritanchor[from=circle]{north west}
  \inheritanchor[from=circle]{north east}
  \inheritanchor[from=circle]{center}
  \inheritanchor[from=circle]{west}
  \inheritanchor[from=circle]{east}
  \inheritanchor[from=circle]{mid}
  \inheritanchor[from=circle]{mid west}
  \inheritanchor[from=circle]{mid east}
  \inheritanchor[from=circle]{base}
  \inheritanchor[from=circle]{base west}
  \inheritanchor[from=circle]{base east}
  \inheritanchor[from=circle]{south}
  \inheritanchor[from=circle]{south west}
  \inheritanchor[from=circle]{south east}
  \inheritbackgroundpath[from=circle]
  \foregroundpath{
    \centerpoint%
    \pgf@xc=\pgf@x%
    \pgf@yc=\pgf@y%
    \pgfutil@tempdima=\radius%
    \pgfmathsetlength{\pgf@xb}{\pgfkeysvalueof{/pgf/outer xsep}}%
    \pgfmathsetlength{\pgf@yb}{\pgfkeysvalueof{/pgf/outer ysep}}%
    \ifdim\pgf@xb<\pgf@yb%
      \advance\pgfutil@tempdima by-\pgf@yb%
    \else%
      \advance\pgfutil@tempdima by-\pgf@xb%
    \fi%
    \pgfpathmoveto{\pgfpointadd{\pgfqpoint{\pgf@xc}{\pgf@yc}}{\pgfqpoint{\pgfutil@tempdima}{0pt}}}
    \pgfpathlineto{\pgfpointadd{\pgfqpoint{\pgf@xc}{\pgf@yc}}{\pgfqpoint{-\pgfutil@tempdima}{0pt}}}
    \pgfpathmoveto{\pgfpointadd{\pgfqpoint{\pgf@xc}{\pgf@yc}}{\pgfqpoint{0pt}{-\pgfutil@tempdima}}}
    \pgfpathlineto{\pgfpointadd{\pgfqpoint{\pgf@xc}{\pgf@yc}}{\pgfqpoint{0pt}{\pgfutil@tempdima}}}
    \pgfsetarrowsstart{}
    \pgfsetarrowsend{}
  }
}
\makeatother

\newcommand*{\defsym}{\triangleq}
\newenvironment{lalign}[1][t]{\begin{array}[#1]{@{}>{{}}l@{}}}{\end{array}}

\newcommand{\N}{\mathbb{N}}
\newcommand{\R}{\mathbb{R}}
\newcommand{\Rext}{\R^{+\infty}}
\newcommand{\kost}{\mathbb{K}}
\newcommand{\kostf}{\mathbb{K}_{\typestyle{f}}}
\newcommand{\keg}{\ensuremath{\mathcal{K}}}
\newcommand{\Rpos}{\R^{+}}

\newcommand{\NF}[2][]{{\mathsf{Term}_{#1}(#2)}}
\newcommand{\MDistr}{\mathsf{MD}}
\newcommand{\LDistr}[1][]{\mathsf{LimDist}_{#1}}

\newcommand{\ecost}[2][]{\mathsf{ecost}_{#1}\left(#2\right)}
\newcommandx{\evalue}[4][1=,2=\keg]{\mathsf{evalue}_{#1,#2}\left(#3\right)(#4)}
\newcommand{\nf}[1][]{\mathsf{nf}_{#1}}

\newcommand{\lmulti}{\{}
\newcommand{\rmulti}{\}}
\newarrow{\rew}{-{Stealth[]}}
\newarrow{\drew}{-{Stealth[sep=-1pt] . Stealth[sep=-5pt]}}

\newcommand{\Distr}{\mathsf{D}}
\newcommand{\supp}{\mathsf{supp}}

\newcommand{\Var}{\textbf{Var}}
\newcommand{\VAR}{\textbf{VAR}}
\newcommand{\Terms}{\mathbf{Terms}}
\newcommand{\Values}{\mathbf{Values}}
\newcommand{\cTerms}{\mathbf{CTerms}}
\newcommand{\cValues}{\mathbf{CValues}}
\newcommand{\EC}{\mathbf{Evaluation\ Contexts}}
\newcommand{\st}{t} 
\newcommand{\sts}{s} 
\newcommand{\sv}{v} 
\newcommand{\sx}{x} 
\newcommand{\sy}{y} 
\renewcommand{\sf}{f} 

\newcommand{\fun}[1]{\mathtt{#1}}
\newcommand{\app}{\cdot}
\newcommand\tick[1]{#1^\checkmark}
\newcommand{\unary}[2][U]{\mathsf{#1}\,#2}
\newcommand\meas[1]{\mathsf{meas}\,#1}

\newcommand{\inj}{\mathsf{inj}}
\newtcbox{\oldbox}[1][gray]{on line,
arc=2pt,colback=#1!35!red,colframe=#1!50!black,
before upper={\rule[-3pt]{0pt}{10pt}},boxrule=0pt,
boxsep=0pt,left=1pt,right=1pt,top=0pt,bottom=0pt}

\newtcbox{\romanobox}[1][gray]{on line,
arc=7pt,colback=#1!35!white,colframe=#1!50!black,
before upper={\rule[-3pt]{0pt}{10pt}},boxrule=0pt,
boxsep=0pt,left=1pt,right=0pt,top=1pt,bottom=1pt}

\newtcbox{\romanosmallbox}[1][gray]{on line,
arc=2pt,colback=#1!35!white,colframe=#1!50!black,
before upper={\rule[-3pt]{0pt}{10pt}},boxrule=0pt,
boxsep=0pt,left=1pt,right=1pt,top=0pt,bottom=0pt}

\newcommand{\block}[2][t]{\begin{array}[#1]{@{}l@{}}#2\end{array}}

\newcommand{\espace}{\,}
\newcommand{\eindent}{\quad}
\NewDocumentCommand{\caseof}{s s m O{\cons(\vec \sx)} m O{\sy} m o}{
  \block{%
    \mathsf{case}\espace#3\espace\mathsf{of}
    \IfBooleanTF {#1} {
      \IfBooleanTF {#2}{
        \espace\{\\\eindent\hspace{-5pt}
        \begin{array}[t]{@{}c@{\,}r@{{}\mapsto{}}l@{}}
          & #4 & #5 \\
          \mid & #6 & #7\espace\}\IfNoValueF{#8}{\espace#8}
        \end{array}
      }{
        \left\{
          \begin{array}{@{}r@{{}\mapsto{}}l@{}}
            #4 & #5 \\
            #6 & #7
          \end{array}
        \right\} \IfNoValueF{#8}{\espace#8}
      }
    }{
      \espace\{\,#4 \mapsto #5 \mid #6 \mapsto #7\,\} \IfNoValueF{#8}{\espace#8}

    }
  }
}

\NewDocumentCommand{\pparg}{m}{\espace #1}
\NewDocumentCommand{\lam}{s >{\SplitList{,}}m m}{
  \block{
    \lambda \ProcessList {#2} { \pparg }.\espace\IfBooleanTF{#1}{\\\eindent}{\espace}\block{#3}
  }
}

\NewDocumentCommand{\letexp}{s O{let} m >{\SplitList{,}}m m}{
  \def\spc{\IfBooleanTF{#1}{\\\eindent}{\espace}}
  \block{
    \mathsf{#2}\espace #3 \ProcessList {#4} { \pparg } =
    \spc\block{#5}
  }
}

\NewDocumentCommand{\letrec}{s m m}{\IfBooleanTF{#1}{\letexp*[letrec]{#2}{#3}}{\letexp[letrec]{#2}{#3}}}

\newcommand{\cons}[1][c]{\texttt{#1}}

\newcommand{\Qbasic}{\mathbb{B}_\qbt}
\newcommand{\Cbasic}{\mathbb{B}_\cbty}
\newcommand{\Basic}{\mathbb{B}}
\newcommand{\typ}{\typestyle{B}}
\newcommand{\ctyp}{\typ_\cbty}
\newcommand{\qtyp}{\typ_\qbt}

\newarrow{\red}{-{Latex[round,open]}}
\newarrow{\dred}{-{Latex[sep=-1pt,round,open] . Latex[sep=-5pt,round,open]}}

\newcommand{\qnil}{\texttt{qnil}}
\newcommand{\qcons}{\texttt{qcons}}
\newcommand{\zero}{\texttt{0}}
\newcommand{\suc}{\texttt{s}}

\newcommand{\Cons}{\texttt{Cons}}
\newcommand{\qtimes}{\mathrel{\tikz[baseline=0pt]{\node[xor,scale=.65,draw=black,semithick,rotate=45,anchor=south]{};}}}
\newcommand{\cadd}{\mathbin{\pmb{\hat{+}}}}

\newcommand{\tplus}{\mathbin{\hat{+}}}
\newcommand{\fplus}{\mathbin{{+}_{\typestyle{f}}}}
\newcommand{\bary}[1]{\mathbin{+_{#1}}}
\newcommand{\cbary}[1]{\mathbin{\pmb{+\!}_{p_0(#1)}}}

\newcommand{\bzero}{\pmb{0}}
\newcommand{\bun}{\pmb{1}}
\newcommand{\breal}[1]{\pmb{#1}}

\newcommand{\cctx}{\Gamma}
\newcommand{\qctx}{\Delta}
\newcommand{\typestyle}[1]{\ensuremath{\mathsf{#1}}}
\newcommand{\ty}{\typestyle{T}}

\renewcommand{\qty}{\typestyle{Q}}
\newcommand{\cty}{\typestyle{C}}

\newcommand{\cbty}{\typestyle{c}}
\newcommand{\qbty}{\typestyle{Q}}

\newcommand{\qbt}{\typestyle{q}}
\newcommand{\Nat}{\texttt{Nat}}
\newcommand{\Out}{\texttt{Out}}

\newcommand{\Qlist}{\texttt{QList}}
\newcommand{\Bool}{\texttt{Bool}}
\newcommand{\Types}{\mathbf{Types}}
\newcommand{\CTypes}{\mathbf{DTypes}}

\newcommand{\TTerms}{{\mathbf{TERMS}}}
\newcommand{\TYPES}{\text{CS Types}}
\newcommand{\FTYPES}{\text{Functional CS Types}}

\newcommand{\TValues}{\textbf{VALUES}}

\newcommand{\te}{T}
\newcommand{\tv}{V}
\newcommand{\tx}{X} 
\newcommand{\tz}{Z} 
\newcommand{\txtwo}{Y}

\newcommand{\tf}{F}

\newcommand{\tk}{K}

\newcommand{\Costfuns}{\mathbf{FUNS}}

\newcommand{\Meas}[2]{\pmb{M}_{\pmb{#1}}#2}

\newcommand{\TTypes}{\mathbf{TYPES}}
\newcommand{\TFTypes}{\mathbf{FTYPES}}
\newcommand{\tty}{\mathsf{S}}
\newcommand{\cs}{\mathsf{CS}}
\newcommand{\fty}{\mathsf{F}}
\newcommand{\kty}{\typestyle{K}}
\newcommand{\realty}{\typestyle{R}^{+ \infty}}
\newcommand{\tctx}{\Theta}

\newcommand{\RTypes}{\dot{\mathbf{RTYPES}}}
\newcommand{\Forms}{\dot{\mathbf{FORMS}}}
\newcommand{\rty}{\dot{\tty}}
\newcommand{\rtb}{{\tt I}}

\newcommand{\skel}[1]{\ulcorner #1 \urcorner}

\newcommand{\rrfty}{\dot{\fty}}
\newcommand{\refine}[3][\tz]{\{ #1 : #2 \mid #3\}}
\newcommand{\refin}[3]{\{ #3 : #1 \mid #2\}}
\newcommand{\rctx}{\dot{\tctx}}
\newcommand{\rsubtype}{\mathrel{\textnormal{\texttt{{<}\!\raisebox{.7pt}{:}}}}}
\newcommand{\ofty}{\ {:}\ }

\newcommand{\form}{\phi}
\newcommand{\formtwo}{\phi'}

\newcommand{\lub}{\bigsqcup}
\newcommand{\tocont}{\to}
\newcommand{\lfp}{\mathsf{lfp}}
\newcommand{\apply}{\mathsf{apply}}

\newcommand{\env}{\rho}
\newcommand{\emptyv}{\emptyset}

\newcommand{\ttpe}[1]{#1^{\bullet}}
\newcommand{\qetv}[1]{#1^{\bullet}}
\newcommand{\transformer}[3]{#1\!\left[{\begin{lalign}[c]#2\end{lalign}}\right]\!\left(\begin{lalign}[c]#3\end{lalign}\right)}

\newcommand{\qet}[2]{\transformer{\mathsf{qet}}{#1}{#2}}

\newcommand{\cont}{K}

\newcommand{\INF}{ \vdash_\cs }
\newcommand{\INFwf}{\Vdash_{\mathsf{wf}}}
\newcommand{\INFr}[1][\kost]{ \Vdash^{#1} }
\newcommand{\INFrr}[1][\Rext]{ \Vdash^{#1} }

\newcommand{\INFb}{ \Vdash_{\mathsf{ST}}}
\newcommand{\INFs}[1][\kost]{ \Vdash_{\mathsf{s}}^{#1} }
\newcommand{\VDash}{
  \mathrel{
    \text{\clipbox{0pt 0pt {.8\width} 0pt}{$\Vdash$}}
    \mkern.9mu
    \text{\adjustbox{width=.87\width,height=\height}{$\vDash$}}
  }
}
\newcommand{\VALID}[1][\kost]{\VDash^{#1} }

\newarrow{\lra}{->}
\newarrow{\disto}{->.>}

\newcommand\tsem[1]{\ensuremath{\llbracket #1\rrbracket}}

\newcommandx{\rtsem}[3][1=,2=\kost]{\ensuremath{\llbracket #3 \rrbracket_{#1}^{#2}}}
\newcommand\ksem[3]{\ensuremath{\tsem{#3}_{#1}^{#2}}}
\newcommand\typesem[1]{\ensuremath{\tsem{#1}^{\kost}}}

\NewDocumentEnvironment{spec}{m o}
{\begingroup\emph{#1}\hfill\IfValueTF{#2}{\fbox{\ensuremath{#2}}}{\,}\newline%
  \setlength{\abovedisplayskip}{0pt}%
  \setlength{\belowdisplayskip}{0pt}%
}
{\endgroup\smallskip\par}

\newcommandx{\mparbox}[3][1=l]{\text{\makebox[#2][#1]{\ensuremath{#3}}}}
\newenvironment{varitemize}
{
\begin{list}{\labelitemi}
{\setlength{\itemsep}{0pt}
 \setlength{\topsep}{0pt}
 \setlength{\parsep}{0pt}
 \setlength{\partopsep}{0pt}
 \setlength{\leftmargin}{15pt}
 \setlength{\rightmargin}{0pt}
 \setlength{\itemindent}{0pt}
 \setlength{\labelsep}{5pt}
 \setlength{\labelwidth}{10pt}
}}
{
 \end{list}
}

\newcounter{numberone}

\newenvironment{varenumerate}
{
\begin{list}{(\arabic{numberone})}
{
  \usecounter{numberone}
  \setlength{\itemsep}{0pt}
  \setlength{\topsep}{0pt}
  \setlength{\parsep}{0pt}
  \setlength{\partopsep}{0pt}
  \setlength{\leftmargin}{15pt}
  \setlength{\rightmargin}{0pt}
  \setlength{\itemindent}{0pt}
  \setlength{\labelsep}{5pt}
  \setlength{\labelwidth}{15pt}
}}
{
\end{list}
}

\makeatletter
\newcommand*{\envskipline}{\hfill\@beginparpenalty=10000}
\newenvironment{enumerateenv}{\hfill\bgroup\@beginparpenalty=10000\begin{varenumerate}}{\end{varenumerate}\egroup}

\makeatother

\newcommand\ite{l}
\newcommand\ITE{\expandafter\uppercase\expandafter{\ite}}

\newcommand\ext[1]{\ensuremath{\overline{#1}}}

\newcommand\tl[1]{\ensuremath{\mathsf{tl}({#1})}}

\newcommand{\Comp}{\mathbb C}
\newcommand{\matrixspace}{\mathcal{M}(\mathcal{H})}
\newcommand{\matrixspacearg}[1]{\mathcal{M}(#1)}

\newcommand{\cointoss}{\ensuremath{\mathsf{cointoss}}}
\newcommand{\cointossK}{\ensuremath{\mathsf{ecost}}}
\newcommand{\COINTOSS}{\fun{COINTOSS}}

\newcommand{\mdist}{\mu}
\newcommand{\mdisttwo}{\nu}
\newcommand{\mdistthree}{o}

\title{Expectation-based Analysis of Higher-Order Quantum Programs}

\author{Martin Avanzini$^1$ \and Alejandro D\'{\i}az-Caro$^{2,3}$ \and Emmanuel Hainry$^2$ \and Romain P\'echoux$^2$}

\date{\small $^1$ INRIA Sophia Antipolis Méditerranée, France\\
  $^2$ Université de Lorraine, CNRS, Inria, LORIA, France\\
  $^3$ Universidad Nacional de Quilmes, Argentina}

\begin{document}
\maketitle 

\begin{abstract}
  The paper introduces a new methodology for studying the expected costs of higher-order quantum programs by proposing a first extension of expectation transformer analysis to the higher-order quantum setting.
Towards that end, we consider an abstract quantum language, an extension of PCF featuring unbounded recursion. The language is universal and admits classical and quantum data, as well as a tick operator to account for costs.
  Our quantum expectation transformer translates such programs into a higher-order non-quantum language, enriched  with a type and operations over so called cost-structures.
  By specializing the cost-structure, this methodology makes it possible to study several quantitative properties of quantum programs, such as average case cost (of a quantum gate), almost-sure termination, or expected values. As a show-case, we adapt a type system, capable of reasoning on upper-bounds.
  \end{abstract}

\section{Introduction}
\subsection{Motivations}
\begin{figure*}[t]
\vspace*{-0.2cm}
\includegraphics[scale=0.72]{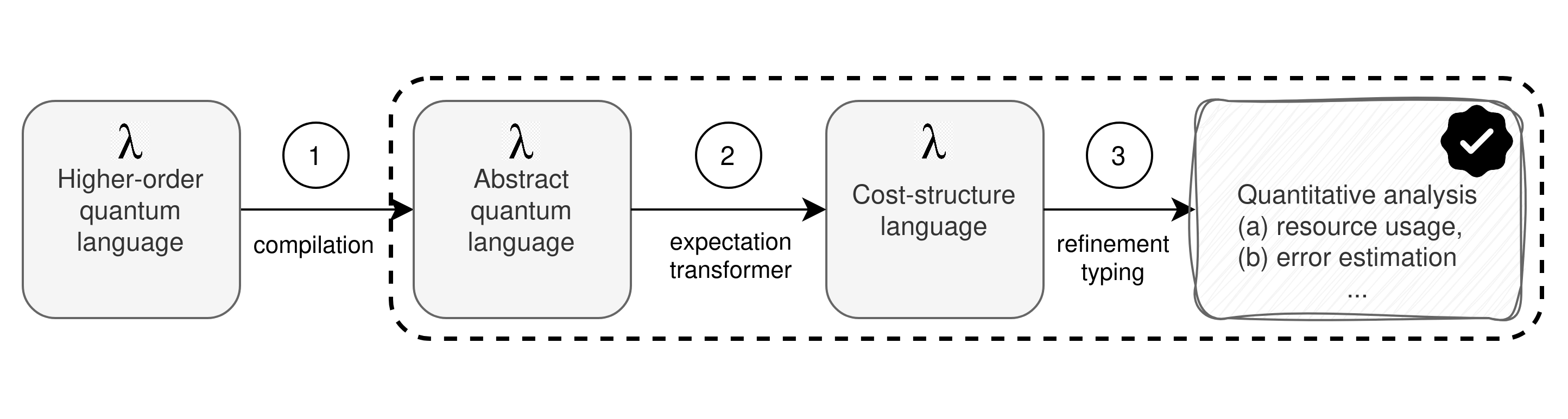}
\caption{Expectation-based Analysis for Higher-Order Quantum Programs}\label{fig:methodology}
\vspace*{-0.2cm}
\end{figure*}
Quantum computing has already demonstrated its ability to solve problems efficiently by improving on the lower bound known for their classical counterparts~\cite{HHL09,S99,G96}.
Popular quantum algorithms solving these problems are typically based on low-level models, such as QRAM~\cite{K96} and quantum circuits~\cite{D89}. Consequently,
there is a need to develop high-level languages, e.g.,  the quantum lambda-calculus~\cite{SV06}
, and to reason about their properties.
The development of such high-level languages cannot be decoupled from low-level considerations in the sense that they must be compilable to a low-level model using a reasonable amount of resources. Indeed, quantum computing is characterized by various physical constraints such as the no-cloning Theorem~\cite{WZ09}, the small number of qubits available on existing machines~\cite{I2024}, or the need to tame quantum decoherence as well as quantum noise~\cite{GKMR14}. The solution to these problems has been to develop static analysis techniques to check that the program satisfies these constraints: non-exhaustively, linear type systems to ensure no-cloning~\cite{AG05}, dependent type systems to control the width and depth of the circuits generated~\cite{CL24,CL25} or weakest preconditions and Hoare logics~\cite{Y11,ZY19} to guarantee quantitative properties of the analyzed programs. In short, the techniques developed must ensure the physicality and feasibility of the programs analyzed.

In recent years, a particular emphasis has been placed on the development of \emph{expectation transformers} techniques~\cite{MM05} for analyzing quantitative data from quantum programs, such as the expected cost, namely the average cost the program experiences along its execution, or the expected value, or the almost-sure termination.  These expectation-based techniques have been successfully exploited to analyze quantitative program  properties and costs, such as the average execution time or the average number of $T$-gates used by a program~\cite{OD20,AMPPZ22,AMPP24,LZBY22}.
 
However, these state-of-the-art techniques are limited to imperative quantum programs as the study of functional programs requires developing a more fine-grained analysis. Hence an open and difficult challenge is to extend expectation transformers to study the resources of \emph{higher-order} \emph{quantum} programs. 

\subsection{Contribution}
We solve the above issue by developing a \emph{compositional} analysis based on expectation transformers to infer expected costs and expected values for higher-order quantum programs. Figure ~\ref{fig:methodology} illustrates the corresponding methodology. 

In Step \Circled{1} of Figure~\ref{fig:methodology}, an arbitrary higher-order quantum program given as input can be compiled to our \emph{abstract quantum language}. 
This abstract language (Figure~\ref{fig:Grammar}) is designed to make resource analysis as simple and feasible as possible: it is simply-typed
with classical control and includes quantum states as first class data.

The abstract language is universal in the sense of~\cite{Barenco95} since its syntax (Figure~\ref{fig:Grammar}) encodes directly the four postulates of quantum mechanics, allowing for the representation of any state vector, the application of arbitrary unitary transformations, and measurements.
By universality, all quantum programming languages, including the ones with complex higher-order
features (e.g., quantum control or quantum pattern-matching), can be considered as input of the analysis and, hence, all the resource-relevant quantum
algorithms can be studied.

Its operational semantics is expressed in terms of probabilistic abstract reduction systems~\cite{BG05} (Figure~\ref{tab:TRS}). The no-cloning Theorem of quantum mechanics is ensured through a linear typing discipline  (Figure~\ref{fig:TS}) using Barber-Plotkin's duality~\cite{BP96}, with a linear and a non-linear arrow.

In Step \Circled{2} of Figure~\ref{fig:methodology},  we define a quantum expectation transformer (Figure~\ref{fig:ect}), that builds upon~\cite{ABL21} and is closely connected to continuation-passing style (CPS) transformation.
The source of the transformation is the abstract quantum language and its target  is a symbolic \emph{cost-structure} language (Figure~\ref{fig:cost:Grammar}), extended with quantum operations for measurements and with cost-structure operations. Cost-structures~\cite{AMPPZ22} correspond to Keimel-Plotkin's Kegelspitzen~\cite{KP17}, roughly speaking domain-theoretic abstract convex sets, extended by a sum operator to account for costs. We define a denotational semantics of the cost-structure language (Figure~\ref{fig:denot}), by giving meaning to recursive definitions in terms of a cost-structures (for example, positive real numbers with infinity).
Under this semantics, the expected cost and expected value of a quantum program can be expressed as the denotation of the cost-structure term (Corollary~\ref{cor:qet}).  As a result, the expectations of a quantum program can be analyzed by looking at the extensional behavior of its transformation (a cost-structure term).

In a last step (\Circled{3}) of Figure~\ref{fig:methodology}, we introduce a refinement type system \cite{FP91,V14} as a method for reasoning about cost-structure terms (Figure~\ref{fig:refine:sub}) similarly to~\cite{KU24}.
Refinement types are a particular form of dependent types, where first-order predicates are used to refine base types.
Soundness (Theorem~\ref{thm:rts}) can be used to derive properties on cost-structure terms: the denotational semantics of a term belongs to the interpretation of its refinement type (Figure~\ref{fig:semforms}) and, as a consequence, sound upper-bounds on the expectations of a quantum program can be derived (Corollary~\ref{cor:qet-refine}). These results crucially rely on an \emph{admissibility} restriction (Definition~\ref{def:admissible}) on the kind of predicates used when dealing with recursive functions.
Restricting to admissible predicates allows us to avoid the usual termination conditions imposed on recursive definitions.
While not all predicates are admissible, this class still includes those of interest for proving upper-bounds on expectations.

We illustrate the main features of the methodology through several examples, putting an emphasis  on the following points: 
\begin{inparaenum}[(i)]
\item The analysis is expressive enough so that the functional versions of the imperative programs of~\cite{AMPPZ22} can be treated, as illustrated by Example~\ref{ss:cointoss}. 
\item In a similar way, all probabilistic programs of ~\cite{ABL21} can be treated through their quantum equivalent. 
\item The analysis applies to higher-order quantum computations, such as the Quantum Walk, as illustrated by Example~\ref{ss:qwalk}.
  \item The analysis can deal with error probability of quantum algorithms, as illustrated by Grover's algorithm in Example~\ref{ss:grover}. 
\item The analysis can be integrated as part of a refinement type system, hence paving the way towards automation.
\end{inparaenum}

\subsection{Related work}
The research pursued in this paper is in the vein of \emph{predicate transformers}, introduced in~\cite{D76,K85} for reasoning about the semantics of
imperative programs. A predicate transformer maps each program statement to a function between two predicates on the program
state space, transforming postconditions (the output predicate) to preconditions (the input predicate).
Consequently, they can be seen as a reformulation of Floyd-Hoare logic~\cite{H69}.
They have been extended to the probabilistic setting, leading to the notion of \emph{expectation
transformers}~\cite{MM05}, where expectations are substituted to predicates.
Expectation transformers have been used to reason about expected costs~\cite{NCH18,AMS20}, expected runtime~\cite{KKMO16,BKKMV23}, as well as expected values~\cite{KK17}.
They also have been extended to deal with quantum imperative programs, leading to the notion of \emph{quantum expectation transformers}~\cite{AMPPZ22,AMPP24,LZBY22}.
In the same vein,~\cite{AMS23} has developed an automated analysis for computing the expected values of probabilistic imperative programs with first-order (recursive) procedures. The paper \cite{ABL21}, from which our paper draws much of its inspiration, has extended the methodology to analyze the expected costs of higher-order probabilistic programs using a Continuation-Passing Style transformation. This approach has been extended in~\cite{KU24} in the form of a refinement type system to deal with expectations of higher-order probabilistic programs.

Alternative logical and type-based approaches have been used in the setting of higher-order randomized programs~\cite{ABG+21,SAB+19,LG17}. By their inherent nature, these techniques contain restrictions which prevent them from computing arbitrary bounds. With regard to quantum languages, several approaches have been developed to control the resources (width, depth) of a circuit generated by a program using dependent types~\cite{CL24,CL25} or to characterize quantum polynomial time through several restrictions~\cite{HPS23,LZ10,Y20}. But these approaches cannot be transposed to compute expectations.

\subsection{A bird's eye view on our analysis}\label{ss:bev}
Consider the $\cointoss$ program below as an illustrative example of our methodology. Section~\ref{sec:examples} discusses richer examples involving genuine quantum data:
\begin{flalign*}
  ~~& \cointoss : \qbty \multimap \qbty & \\[-2mm]
  & \letrec{\cointoss}{\sx}{
    \caseof**{\tick{(\meas \sx)}}
    [\inj_0(\sx_0)]{\sx_0}
    [\inj_1(\sx_1)]{\cointoss \app \unary[H]{\sx_1}}
    }
\end{flalign*}
Applied on an input quantum state $\ket{\psi}$,
the recursively defined quantum procedure $\cointoss$ performs
a quantum measurement on the first qubit of $\ket{\psi}$, and,
depending on the outcome of the measurement encoded by the constructor $\inj_0$ or $\inj_1$, either returns
the post collapsed quantum state, or recursively iterates after applying the Hadamard gate.
Measurements and unary gates such as the Hadamard gate act on the first qubit, but
can be applied to arbitrary qubits using swap gates.
The tick construct $\tick{(-)}$ is used to explicitly impose a notion of cost; here, the number of measurements performed.
Finally, note that $\cointoss$ is attributed a linear type
$\qbty \multimap \qbty$, thereby ensuring that the no-cloning theorem
of quantum mechanics is adhered to.

Measurements are modeled as probabilistic operations. On a quantum state $\ket{\psi}$,
measurement returns $\inj_i(M_i\ket{\psi})$ with probability $p_i(\ket{\psi})$ ($i \in \{0,1\}$),
with $M_i\ket{\psi}$ and $p_i({\ket{\psi}})$ the post collapsed quantum state and associated probability of measuring outcome $i$, respectively.
For example, if applied to the quantum state $\ket{\psi} \triangleq \sfrac{1}{2}\ket{001} + \sfrac{1}{\sqrt{2}}\ket{011} + \sfrac{1}{2} \ket{100}$,
$\cointoss$ returns  $M_0\ket{\psi} = \sfrac{\sqrt{2}}{\sqrt{3}}\ket{001} + \sfrac{1}{\sqrt{3}}\ket{011}$
with probability $p_0({\ket{\psi}})= \sfrac{3}{4}$,
or recursively iterates on $H M_1\ket{\psi} = H \ket{100} = \sfrac{1}{\sqrt{2}}\ket{000}-\sfrac{1}{\sqrt{2}}\ket{100}$
with probability $\sfrac{1}{4}$.

In the worst-case, the cost of $\cointoss$---the number of measurements performed---is unbounded.
However the probability of $\cointoss$ performing infinitely many measurements, precisely the probability of an infinite sequence of measurements with outcome $1$, is zero (the program is almost-surely terminating). The expected number of measurements performed
by this program is thus a quantity more informative than a worst-case bound.
The main concern of this work is to be able to reason formally on such expectations.

To this end, we proceed in two steps.
First, the expectation transformer defined in this paper translates a term of the abstract quantum language to a cost-structure term that
precisely computes the cost of the quantum term. For instance,
on $\cointoss$ it produces
\begin{align*}
  & \cointossK : \qbty \Rightarrow \realty \\[-1mm]
  & \letrec{\cointossK}{\sx}{1 +p_{1}(\sx)\times (\cointossK \app \unary[H]{(M_1{\sx})})}
\end{align*}
where $\realty$ is the type of non-negative reals completed with $\infty$ and where $\Rightarrow$ is a non-linear arrow.
There is no linearity check on this intermediate term as it just denotes
an expression referring to the cost of $\cointoss$ and has no computational meaning.
To see how this term refers to the cost of $\cointoss$,
reconsider the state $\ket{\psi} = (\sfrac{1}{2}\ket{001} + \sfrac{1}{\sqrt{2}}\ket{011} + \sfrac{1}{2} \ket{100})$
and observe
\begin{align*}
      \cointossK \app \ket{\psi} 
      & = 1 + p_1(\ket{\psi})\times (\cointossK \app \unary[H]{(M_1{\ket{\psi}})}) \\
      & = 1 + \sfrac14 \times (\cointossK \app (\sfrac1{\sqrt{2}}\ket{000}+\sfrac1{\sqrt{2}}\ket{100}))\\
      & = 1 + \sfrac14 \times (1+ \sfrac12 \times (\cointossK \app (\sfrac1{\sqrt{2}}\ket{000}+\sfrac1{\sqrt{2}}\ket{100}))\\
      & = \dots  = 1 + \frac14 \sum_{n=0}^{\infty} {\frac1{2^n}} = \frac32
\end{align*}
Thus $\cointossK\app\ket{\psi}$ converges to $\frac32$ which is exactly to the expected number of measurements performed by program $\cointoss$ on input $\ket{\psi}$.
It can be generalized to any input as $\cointossK = \lambda \sx . 1 + 2 p_{1}(\sx)$.

While $\cointossK$ precisely reflects the expected cost of the program $\cointoss$, it is not informative as it generally requires the computation of some limits. We rather are interested in closed form expression such as $ \lambda \sx . 1 + 2 p_{1}(\sx)$. Towards that end, as a second step, we adapt a \emph{refinement type system} that allows us to derive upper bounds on such  closed form cost-structure expressions. In refinement types, the idea is to attach logical information to base types, enabling reasoning about program properties at the level of types. On the $\cointoss$ example, we can get rid of the fixpoint computation by checking a typing judgment of the kind
\[
  \cdot  \INFr[] \cointossK : (\sx: \qbty) \Rightarrow \refine{\realty}{\tz \leq  1 + 2 p_{1}(\sx)},
\]
meaning that {\cointossK} can be typed under the empty context $\cdot$ as mapping a quantum state $\sx$ to an element of the set of non-negative or infinite real numbers $\tz$ that satisfy the formula $\tz \leq  1 + 2 p_{1}(\sx)$.
Hence, by providing an approximation (upper-bound), the type system allows the programmer to formally check quantitative properties of quantum programs.

We have demonstrated the interest of our approach in the context of expected cost analysis, but as we shall see later, this approach can also be used in the context of expectation based functional analysis.

\section{Preliminaries}

In this section, we introduce a notion of cost-structure, based on Kegelspitzen~\cite{KP17}, that will be used as a semantic domain for programs.
Then we define a notion of weighted probabilistic abstract reduction systems, a probabilistic abstract model used to describe the operational semantics of quantum programs.
This model, combined with cost-structures allows us to define the notions of expected cost and expected value.

\subsection{Kegelspitzen and Cost-structures}\label{ss:k}
Let us recall some basic
notions on domain theory, using terminology
from~\cite[Section 8]{Winskel:93}.  A partial order $(\mathcal{D},\sqsubseteq)$
is called a $\omega$-\emph{complete partial order} ($\omega$-\emph{cpo}) if any \emph{$\omega$-chain} $d_0
\sqsubseteq d_1 \sqsubseteq \cdots$ has a least upper bound $\lub_{n
  \in \N} d_n$ in $\mathcal{D}$.
Whenever $\mathcal{D}$ has a least element, it will be denoted by $\bot$.
Relevant examples of $\omega$-cpos are \emph{discrete cpos} where the order is
the identity relation,
$(\Rext,\leq)$ with $\Rext \triangleq \Rpos \cup \{\infty\}$ the non-negative real numbers extended by the top element $\infty$,
and the function space $([\mathcal{D} \tocont \mathcal{E}],\sqsubseteq)$ of \emph{continuous functions}
between $\omega$-cpos $\mathcal{D}$ and $\mathcal{E}$, ordered point-wise, i.e., $f \sqsubseteq g$
iff for all $d \in \mathcal{D}$, $f(d) \sqsubseteq g(d)$.  A function $f :
\mathcal{D} \to \mathcal{E}$ is continuous if it is \emph{monotone} ($d \sqsubseteq e$
implies $f(d) \sqsubseteq f(e)$) and $\lub_{n \in \N} f(d_n) =
f(\lub_{n \in \N} d_n)$ for all $\omega$-chains $d_0 \sqsubseteq d_1
\sqsubseteq \cdots$.

\begin{definition}
A barycentric algebra $\mathcal{A}$ is a set equipped with a relation $\cdot \bary{\cdot} \cdot \subseteq \mathcal{A} \times [0, 1] \times \mathcal{A}$ satisfying the following laws for all $a, b,c \in \mathcal{A}$ and  $r,s\in[0, 1]$:
\begin{align*}
a\bary1 b &= a & a\bary{r} b &= b\bary{1-r}a\\
a\bary{r} a &= a &
(a\bary{r} b)\bary{s} c &= a \bary{rs} (b\bary{\frac{s-rs}{1-rs}}c)
\end{align*}
provided $rs\neq 1$. The relation $\bary{r}$ is called a barycentric operation.
\end{definition}

\begin{definition}
  A pointed barycentric algebra is a barycentric algebra $\mathcal{A}$ equipped with a distinguished element $\bot \in \mathcal{A}$.
\end{definition}
For a pointed barycentric algebra, we can define scalar multiplication for $r\in[0,1]$ and $a\in \mathcal{A}$ by $r\cdot a \triangleq a\bary{r}\bot$.
The barycentric operations can also be extended to define a convex sum operation for $a_i\in \mathcal{A}$ and $r_i \in [0,1]$ such that
$\sum_{i=1}^{n} r_i\cdot \leq 1$:
$$\sum_{i=1}^{n} r_i\, a_i \triangleq\begin{cases}
\bot & \text{if }n=0\\
 a_n  & \text{if }r_n=1\\
a_n \bary{r_n}\sum_{i=1}^{n-1}\frac{r_i}{1-r_n}\cdot a_i & \text{otherwise}
\end{cases}$$

An $\omega$-Kegelspitze~\cite{KP17} is a structure that combines an $\omega$-cpo with continuous barycentric operations, thus giving a framework to compute expectations.

\begin{definition}
  An $\omega$-Kegelspitze $\keg$ is a pointed barycentric algebra equipped with an $\omega$-cpo such that the scalar multiplication is continuous in both its arguments and for all $r\in[0, 1]$, the barycentric operation $\bary{r}$ is continuous in both arguments.
\end{definition}
As all the Kegelspitzen we use are $\omega$-Kegelspitzen, from now on, we will call them Kegelspitze for short.
Note that denoting the distinguished element by $\bot$ is coherent with its use as the least element since, in a Kegelspitze, by continuity of the scalar multiplication, for all $a$, we have $0\cdot a \sqsubseteq 1\cdot a$ hence $\bot \sqsubseteq a$.

Finally, a cost-structure~\cite{AMPPZ22} is a Kegelspitze equipped with an operation for adding non-negative costs.
\begin{definition}\label{def:cs}
  A cost-structure $\kost$ is defined by $\kost \triangleq (\keg, \tplus)$ for some $\omega$-Kegelspitze $\keg$
  equipped with an operator $\tplus:\Rext\times \keg \to \keg$ that is continuous and satisfies the following laws for all $a,b \in \Rext$, $c,d \in \keg$, and $r \in[0,1]$:
$$0 \tplus c  = c  \qquad a \tplus (b \tplus c)  = (a + b) \tplus c $$
 $$   (a \tplus c) \bary{r} (b \tplus d)  = (a \bary{r} b) \tplus (c \bary{r} d) $$
\end{definition}

\begin{example}
  For $r\in[0, 1]$, define $a\bary{r}b \triangleq r\cdot a+(1-r)\cdot b$ which is a barycentric operation.
  Then both, $[0,1]$ and $\Rext$ with $\bot\triangleq0$ and barycentric operation $\bary{r}$
  form an $\omega$-Kegelspitze under the standard ordering $\leq$.
  The induced scalar multiplication $\cdot$ corresponds to standard multiplication.

  The Kegelspitze $\Rext$ can be complemented with the standard addition as operator $\tplus$ to form a cost-structure $(\Rext, +)$.
  Any Kegelspitze $\keg$ can be extended to a cost-structure $(\keg, \fplus)$
  using the forgetful addition $\fplus$ defined by $a \fplus b \triangleq b$.
\end{example}

\subsection{Probabilistic Abstract Reduction System}\label{s:pars}

\emph{Probabilistic Abstract Reduction Systems} (PARSs) were introduced as a means to study reduction systems with probabilistic behavior~\cite{BG05}. We will use PARSs as a means to endow
an operational semantics to our quantum language. We quickly recap central definitions.

Discrete \emph{subdistributions} $\delta$ over a set $A$ are functions $\delta : A \to [0,1] $ with countable support $\supp(\delta) \triangleq \{ a \in A \mid \delta(a) \neq 0 \}$, mapping each element $a \in A$ to a probability $\delta(a)$ such that $|\delta| \triangleq \sum_{a \in \supp(\delta)} \delta(a) \leq 1$. In the particular case where $|\delta | = 1$, $\delta$ is simply called a \emph{distribution}. Any (sub)distribution $\delta$ can be written as $\lmulti a^{\delta(a)} \rmulti_{a \in \supp(\delta)}$. The set of subdistributions over $A$, denoted by $\Distr(A)$, is closed under denumerable convex combinations $\sum_{i \in I} p_i \cdot \delta_i \triangleq \lambda a.\sum_{i \in I} p_i\cdot \delta_i(a)$, with $p_i \in [0,1]$ and $\sum_{i} p_{i \in I} \leq 1$.

A \emph{deterministic weighted Probabilistic Abstract Reduction System} (PARS) on $A$ is a deterministic ternary relation $\cdot \rew{\cdot} \cdot \subseteq A \times \Rpos \times  \Distr(A)$ such that if $a \rew{c} \delta$ holds, then $a \in A$ reduces with cost $c \in \Rpos$ to the subdistribution $\delta \in \Distr(A)$.
When there is no rule $a \rew{c} \delta$, we write $a \not\rew{}$ and the object $a \in A$ is called \emph{terminal}.

Every deterministic weighted PARS over $A$ can be lifted to a relation over distributions: $\cdot \drew{\cdot} \cdot  \subseteq \Distr(A) \times \Rpos \times \Distr(A)$ using the following rules:
\newcommand{\TOP}{\rule{0pt}{1ex}}
\[\scalebox{0.9}{
  \Infer{
    \vphantom{\drew{{\sum}_{i \in I} p_i \cdot c_i}{}}
    \{ a^1 \} \drew{0} \{ a^1 \}
  }{
    a \not\rew
  }
  }
  \hspace*{2mm}
  \scalebox{0.9}{
  \Infer{
    \vphantom{\drew{{\sum}_{i \in I} p_i \cdot c_i}{}}
    \{a^1\} \drew{c} \delta
  }{
    a \rew{c} \delta
  }
  }
  \hspace*{2mm}
  \scalebox{0.9}{
  \Infer{
    {\sum}_{i \in I} p_i \cdot \delta_i \drew{{\sum}_{i \in I} p_i \cdot c_i}{}  {\sum}_{i \in I} p_i \cdot \delta_i'
  }{
    \forall i \in I,\ \delta_i \drew{c_i} \delta_i'
  }
  }
\]
If $\delta \drew{c} \epsilon$, then the subdistribution $\epsilon$ is obtained by rewriting all the elements in
the support of $\delta$. The weight $c$ signifies the (expected) cost of this reduction step.
Notice that since $\rew{}$ is deterministic, so is $\drew{}$, in the sense that
${\epsilon_1 \drew<{c_1} \delta \drew{c_2} \epsilon_2}$ implies $\epsilon_1 = \epsilon_2$ and $c_1 = c_2$.
Thus, in particular, for every $\delta \in \Distr(A)$ there is precisely one (infinite) reduction sequence
\[
  s: \delta = \delta_0 \drew{c_1} \delta_1 \drew{c_2} \delta_2 \drew{c_3} \cdots
  .
\]
The distributions $\delta_n$ give the $n$-step reduct of $\delta$,
the sum $\sum_{i=1}^{n} c_i$ signifies the expected cost to reach $\delta_n$ in $n$ steps.
As normal forms $a \not\rew$ remain stable under reductions $\drew{}$, the sequence $s$ can be seen to converge to a \emph{subdistribution over terminal objects},
with an \emph{overall expected cost} $\sum_{i=1}^{\infty} c_i$.

The \emph{expected cost function} $\ecost[\drew]{\cdot} : A \to \Rext$ is defined by:
\[
  \ecost[\drew]{a} \triangleq \sup \{ \sum_{i = 0}^n c_n \mid \{a^1\} = \delta_0 \drew{c_1} \cdots \drew{c_n} \delta_n \}\,.
\]
The \emph{normal form distribution} $\nf[\drew](a) \in \Distr(A)$ is defined by
\[
  \nf[\drew](a) \triangleq \sup \{ \NF[\drew]{\delta_n} \mid \{a^1\} = \delta_0 \drew{c_1} \cdots \drew{c_n} \delta_n \}\,,
\]
for $\NF[\drew]{\delta} \in \Distr(A)$ the subdistribution of normal forms in $\delta$, i.e.,
$\NF[\drew]{\delta}(a) = \delta(a)$ if $a \not \rew$, and $\NF[\drew]{\delta}(a) = 0$ otherwise.
Here, the supremum in the definition of $\nf[\drew](\delta)$ is taken with respect to the point-wise ordering.
This supremum always exists, as the sequence $(\NF[\drew]{\delta_n})_{n \in \N}$ is non-decreasing.
Finally, for a given Kegelspitze $\keg$, we can define
\emph{the expected value function} $\evalue[\drew]{\cdot}{\cdot}:A\to(A\to\keg) \to \keg$ by
\[
  \textstyle
  \evalue[\drew]{a}{f} \triangleq \sum_{b \in A} \nf[\drew](a)(b) \cdot f(b).
\]
In the particular case where $\keg = [0,1]$, $\evalue[\drew]{a}{P}$ gives the probability that
predicate $P: A \to \{0,1\}$ holds when $a$ has been fully reduced.

\section{Abstract Quantum Language}\label{ss:ts}
\subsection{Syntax}

\begin{figure}[t]
  \centering
  \columnwidth=\linewidth
  \begin{bnf}
    \Terms \ni \st
    & \sx | \lambda \sx.\st | \st_0 \app \st_1 & $\lambda$-calculus
    \\
    & \ket\psi | \unary{\st} | \meas \st | \st_0 \qtimes \st_1 & Quantum data
    \\
    & \cons(\vec{\st_0};\vec{\st_1} ) & Constructor
    \\
    & \caseof {\st}[\cons(\vec\sv_0;\vec\sv_1)]{\st_0}{\st_1} & Destr.
    \\
    & \letrec{\sf}{\sx}{\st} & Recursion
    \\
    & \tick \st  & Cost
  \end{bnf}
  \vspace*{0.2cm}
  \begin{bnf}
    \Values \ni \sv
    & \sx
    | \lambda \sx.\st
    | \ket{\psi}
    | \cons(\vec \sv_0;\vec \sv_1)
    | \letrec{f}{\sx}{\st}
  \end{bnf}
  \caption{Abstract Quantum Language}
  \label{fig:Grammar}
\end{figure}

We consider a  functional quantum programming language with classical control  in the style of~\cite{DiazcaroAPLAS17}.
It can be viewed as an adaptation of the probabilistic functional language of~\cite{ABL21} to the quantum paradigm. The corresponding sets of terms $\Terms$ and values $\Values$ are defined by the grammar of Figure~\ref{fig:Grammar},
where
variables $\sx, \sf $ belong to a fixed countable set $\Var$.
The set $\Terms$ includes:
\begin{varitemize}
  \item the standard constructs of the lambda calculus;
  \item  four \emph{quantum constructs}: \emph{quantum states} $\ket{\psi}$, \emph{unitary applications} $\unary{\st}$, applying the \emph{unitary transformation} $\unary$ to a term $\st$, 
\emph{measurement} $\meas \st$, and \emph{tensoring} $\st_0 \qtimes \st_1$;
  \item \emph{constructor applications} $\cons(\vec{\st_0};\vec{\st_1})$ and the corresponding \emph{destructors}
$\caseof{\st}[\cons(\vec{\sx}_0;\vec{\sx}_1)]{\st_0}[\sy]{\st_1}$;
  \item the \emph{letrec} construct $\letrec fxt$ to define recursive functions;
  \item the \emph{tick} construct $\tick t$ to account for costs.
\end{varitemize}
Let $\mathsf{fv}(\st)$ be the set of free variables in $\st$.
We denote by $\cTerms$ and $\cValues$ the sets of closed terms and values, respectively.

Each quantum state  $\ket \psi$ comes with a length $l(\ket \psi) \in \mathbb N$ denoting the number of qubits represented by $\psi$. A state $\ket \psi$ is a vector of norm $1$ (i.e., $\|\ket{\psi}\|=1$) in the Hilbert space $
\Comp^{2^{l(\ket \psi)}}$.

Each unitary transformation $\unary$ comes with a length $l(\unary) \in \mathbb N$  denoting the number of qubits on which $\unary$ can be applied.  For a given Hilbert space $\mathcal H$ of dimension $n$, let $\matrixspace$ be the set of complex square matrices acting on $\mathcal H$.
Given $M \in \matrixspace$, $M^\dagger$ denotes the transpose conjugate of $M$, and $I_{n}$ (or simply $I$, when $n$ is clear from context) denotes the identity matrix over $\matrixspace$.
A \emph{unitary operator} is a matrix
 $U \in \matrixspace$ such that $UU^\dagger = U^\dagger U = I$.
Each symbol $\unary$ is associated with a unitary operator in $\matrixspacearg{\Comp^{2^{l(\unary)}}}$.

In a constructor application $ \cons(\vec{\st_0};\vec{\st_1} ) $, the classical and quantum operands are separated by a semi-colon. This distinction will be used by the typing discipline described in Section~\ref{ss:ts},
to enforce the no-cloning theorem.
We suppose the language includes as constructors at least injections $\inj_0$ and $\inj_1$,
that will be used to model the outcome of a quantum measurement.

\begin{figure}[t]
  \columnwidth=\linewidth
  \begin{center}
  \begin{align*}
   (\lambda \sx.\st)\app \sv &\lra{0} \{\st[\sv/\sx]^1\}\\
    \unary\ket{\psi} &\lra{0} \{(\ext{U}\,\ket{\psi})^1\} \\
    \meas{\ket\psi} &\lra{0}\{\inj_i(;\!M_i\ket{\psi})^{p_i \ket \psi}\}_{i\in\{0,1\}}\\
   \ket{\psi_0}\qtimes\ket{\psi_1}&\lra{0}\{(\ket{\psi_0}\otimes\ket{\psi_1})^1\}  \\
    \caseof*{\sv}[\cons(\vec{\sx}_0;\vec{\sx}_1)]{\st_0}[\sy]{\st_1}\!
     & \lra{0} \!\!
       \begin{cases}
         \!\{\st_0[\vec \sv_0/\vec \sx_0,\vec \sv_1/\vec \sx_1]^1\} & \!\!\text{if $\sv=\cons(\vec \sv_0;\vec \sv_1)$} \\
         \!\{\st_1[\sv/\sy]^1\} & \!\!\text{otherwise}
       \end{cases}
    \\
    (\letrec \sf\sx\st) \app \sv &\lra{0} \{\st[\letrec \sf\sx\st/\sf, \sv/\sx]^1\}\\
    \tick \st&\lra{1} \{\st^1\}
    \\
    E[\st]&\lra{c} \{E[\sts]^{\delta(\sts)} \}_{\sts \in \supp(\delta)} \quad \text{if } \st \lra{c} \delta
  \end{align*}
  \end{center}

\raggedright
where $\EC$ $E$ are defined by:
  \begin{align*}
E \ ::=\ & []\ |\ E \app \sv\ |\ \st \app E\ |\  \unary{E} \ |\ \meas{E}\  |\ E \qtimes \sv \   |\ \st \qtimes E \ | \ \cons(\vec{\st}_0;\vec{\st}_1,E,\vec{\sv}_1) \\
  &   \cons(\vec{\st}_0,E,\vec{\sv}_0;\vec{\sv}_1)\  |\  \caseof{E}{\st_0}{\st_1} 
\end{align*}

  \caption{PARS of the Abstract Quantum Language}
  \label{tab:TRS}
\end{figure}

\subsection{Operational Semantics}\label{ss:os}
In Figure~\ref{tab:TRS}, we endow our quantum language with a call-by-value operational semantics,
given in terms of the (deterministic) PARS
\[
  \cdot \lra{\cdot} \cdot \subseteq \cTerms \times \Rpos \times  \Distr(\cTerms)
\]
on ground terms.
We briefly comment the rules of Figure~\ref{tab:TRS} for the quantum and non-standard constructs.

The semantics of quantum gate application is defined using the extension $\ext{U}$.
Let
$\ext{U}_n$ denote the extension of a unitary $U$ to a unitary of length $n \in \mathbb{N}$
and, by convention, the identity when $n \leq l(U)$.
Formally, $\ext{U}_n \defsym U \otimes I_{n - l(U)}$ for $l(U) \leq n$, and otherwise $\ext{U}_n \defsym I_n$.
We write $\ext{U}$ for $\ext{U}_{l(\ket{\psi})}$, when applied to a quantum state $\ket{\psi}$.

The rule for measurement $\meas$is the only probabilistic rule. It reduces to a probability distribution consisting of (at most) two terms, corresponding to the two possible measurement outcomes.
Measurement is designed to behave as a sampling operator for our soundness proof. This is the reason why these two injections $\inj_i$, for $i \in \{0,1\}$, are used as marks of a measurement outcome.
Here, the probabilities $p_i \ket{\psi}$ and the post-collapse quantum states $M_i \ket{\psi}$, for $i \in \{0,1\}$, are defined as usual:
\begin{align*}
  p_i  \ket{\psi} &\triangleq  tr(\ext{\ketbra{i}} \ketbra{\psi} )
  &
    M_i \ket{\psi} & \triangleq
                     \begin{cases}
                       \frac{\ext{\ketbra{i}}}{\sqrt{p_i \ket  \psi}} \ket{\psi}& \text{if } p_i \ket \psi\neq 0 \\
                       \ket{\psi}   & \text{if } p_i \ket \psi= 0.
                     \end{cases}
\end{align*}
For the sake of having a total definition, we have defined $M_i$ as an identity, when $p_i \ket \psi = 0$.
Note however that in this case $M_i \ket{\psi} = \ket{\psi}$ will not occur in the support of the reduct distribution, as it will be attributed probability $0$.

The rule for $ \tick{}$ is the only rule that accounts for a cost and the last rule closes
the reduction relation under \emph{evaluation contexts}.

We denote by $\disto{}$ the lifting of $\lra{}$ to distributions $\Distr(\cTerms)$.

\subsection{Type System}

  The set of \emph{basic types} $\Basic=\{\typ , \typ_1,\ldots \}$ is defined as the disjoint union of the set
  of \emph{classical basic types} $\Cbasic=\{\ctyp,\ctyp^1,\ldots \}$
  and \emph{quantum basic types} $\Qbasic=\{\qtyp,\qtyp^1,\ldots \}$, the latter including distinguished
  types $\qbty$ of quantum states and $\Out$ for measurement outcomes.
  This distinction will be used to delineate those classical values that do not have to adhere to the no-cloning theorem.

Each type $\typ \in \Basic$ comes with a finite set of constructor symbols $\Cons(\typ)$, each constructor symbol in $\Cons(\typ)$ having a fixed type signature of the shape $\cons :: \ctyp^1 \times \ldots \times \ctyp^n \times \qtyp^1 \times \ldots \times \qtyp^m \to \typ$, with  $m,n \in \mathbb{N}$.  We will sometimes write $\cons :: \vec \ctyp ;\! \vec\qtyp \to \typ$ as a shorthand notation, and also we may write just $\typ$ when $n=m=0$.

For each $\ctyp \in \Cbasic$, the constructor symbols in $\Cons(\ctyp)$ cannot take quantum basic types in their signature, i.e., $m=0$. Hence each $\cons$ in $\Cons(\ctyp)$ has a signature $\cons :: \vec \ctyp;\! \to \ctyp$. 
The type $\Nat$ of numerals such that $\Cons(\Nat) \triangleq \{\suc,\zero\}$, with $\suc:: \Nat; \to \Nat$ and $\zero::\Nat$ is an example of type in $\Cbasic$.

The type $\qbty \in \Qbasic$ has no constructor, i.e.,  $\Cons(\qbty)=\emptyset$.  The type of measurement outcomes $\Out\in \Qbasic$ contains the two injections $\Cons(\Out)\triangleq \{\inj_0,\inj_1\}$ of type signature $\inj_i ::\,; \qbty \to \Out$.
The type $\Qlist$ of quantum lists can be defined $\Cons(\Qlist) \triangleq \{\qnil, \qcons \} $, with $\qnil:: \Qlist$ and $\qcons::\,;\qbty \times \Qlist \to \Qlist$.

\begin{figure*}[t]
  \begin{center}
    \begin{minipage}{.45\textwidth}
      \begin{bnf}
        \CTypes \ni \cty  &    \ctyp | \cty \Rightarrow \ty | \ty \multimap \ty & Duplicable types\\
        \Types \ni \ty  & \qtyp | \cty & Types\\
      \end{bnf}
    \end{minipage}
\qquad
    \begin{minipage}{.4\textwidth}
      \begin{bnf}
        \cctx & \emptyset | \cctx, \sx: \cty & Exponential context \\
        \qctx & \emptyset | \qctx, \sx:\ty &  Affine context
      \end{bnf}
    \end{minipage}
  \end{center}
  \vspace{5mm}

  \def\inferLineSkip{1mm}
  \begin{infers}
    \Infer[sty][axq][ax_c]{\cctx,\sx:\cty;\qctx\vdash \sx:\cty}{}
    &
    \Infer[sty][ax][ax]{\cctx;\qctx,\sx:\ty\vdash \sx:\ty}{}
    &
    \Infer[sty][lini][\multimap_i]{\cctx;\qctx\vdash\lambda \sx.\st: \ty \multimap \ty'}{\cctx;\qctx,\sx:\ty \vdash \st:\ty'}
    &
    \Infer[sty][line][\multimap_e]{\cctx;\qctx_0,\qctx_1\vdash \st_0 \app \st_1:\ty'}{\cctx;\qctx_0\vdash \st_0:\ty \multimap \ty' & \cctx;\qctx_1\vdash \st_1:\ty}
    \\
    \Infer[sty][expi][\Rightarrow_i]{\cctx;\qctx\vdash\lambda \sx.\st:\cty\Rightarrow \ty}{\cctx,\sx:\cty;\qctx\vdash \st:\ty}
    &
    \Infer[sty][expe][\Rightarrow_e]{\cctx;\qctx\vdash \st_0 \app \st_1:\ty}{\cctx;\qctx\vdash \st_0:\cty \Rightarrow \ty & \cctx;\emptyset\vdash \st_1:\cty}
    &
    \Infer[sty][st]{\cctx;\qctx\vdash\ket{\psi}:\qbty}{\strut}
    &
    \Infer[sty][un]{\cctx;\qctx\vdash U\, \st:\qbty}{\cctx;\qctx\vdash \st:\qbty}
    \\
    \Infer[sty][meas]{\cctx;\qctx\vdash\meas{\st}:\Out}{\cctx;\qctx\vdash \st:\qbty}
    &
    \Infer[sty][prod][\otimes]{\cctx;\qctx_0,\qctx_1\vdash \st_0\qtimes \st_1:\qbty}{\cctx;\qctx_0\vdash \st_0:\qbty & \cctx;\qctx_1\vdash \st_1:\qbty}
    &
    \Infer[sty][cons]{\cctx; \vec \qctx_0, \vec \qctx_1 \vdash \cons(\vec\st_0;\vec \st_1):\typ}{\cctx;\vec \qctx_0 \vdash \vec \st_0:\vec \ctyp &   \cctx;\vec \qctx_1 \vdash \vec \st_1: \vec \qtyp  &\cons :: \vec \ctyp ;\! \vec \qtyp \to \typ}
    \\
    \Infer[sty][muche][case]{
      \cctx;\qctx_0,\qctx_1\vdash\caseof{\st}[\cons(\vec{\sx_0}; \vec{\sx_1})]{\st_0}[\sy]{\st_1} : \ty
    }{
      \begin{array}{@{}l@{}}
        \cctx;\qctx_0 \vdash \st:\typ\\
        \cons \!:: \!\vec \ctyp ;\! \vec \qtyp \to \typ\!\!\!\!\!\!\!
      \end{array}
      &
      \begin{array}{@{}r@{}l@{}}
        \cctx ,\vec{\sx}_0:\vec{\ctyp};\qctx_1 ,\vec{\sx_1}:\vec{\qtyp} & {} \vdash \st_0:\ty\\
        \cctx' ,\qctx_1' & {} \vdash \st_1:\ty \\
      \end{array}
      &
      \cctx';\qctx_1' =
      \begin{cases}
        \cctx', \sy\!:\!\typ;\qctx_1' & \text{if $\typ \in \Cbasic$,}\\
        \cctx';\qctx_1', \sy\!:\!\typ & \text{if $\typ \in \Qbasic$.}
      \end{cases}
    }
 &
 \Infer[sty][rec]{\cctx;\emptyset \vdash \letrec \sf\sx\st:\cty }{\cctx,\sf:\cty ;\emptyset \vdash \lambda \sx.\st: \cty }
 &
    \Infer[sty][tick][\checkmark]{\cctx;\qctx\vdash\tick \st:\ty}{\cctx;\qctx\vdash \st:\ty}
  \end{infers}
  \caption{Types and Rules for the Abstract Quantum Language.}
  \label{fig:TS}
\end{figure*}

The set of  \emph{duplicable types} $\CTypes$ contains all types $\cty$, that can be cloned, and is defined in Figure~\ref{fig:TS} as the set of types that are either  classical basic types or functional types. There are two distinct arrow types for abstractions $\ty \multimap \ty'$ and $\cty \Rightarrow \ty$, depending on whether the function takes affine data (including quantum data) of type $\ty$ as input or exponential data (only duplicable data) of type $\cty$ as input.
The set of types $\Types$ extends strictly $\CTypes$ with quantum basic types.

The type discipline of Figure~\ref{fig:TS} is presented in the form of a dual intuitionistic type system à la~Barber and Plotkin~\cite{BP96}.  Typing judgments have the shape
    $\cctx;\qctx \vdash \st : \ty$, for some term $\st \in \Terms$ and some type $\ty$.    Given a sequence of affine contexts $\vec \qctx=\qctx_1,\ldots,\qctx_n$, a sequence of terms $\vec \st=\st_1,\ldots,\st_n$, and a sequence of types $\vec \typ = \typ_1,\ldots,\typ_{n}$, we write $\cctx;\vec \qctx \vdash \vec \st : \vec\ty$ when $\cctx; \qctx_i \vdash \st_i : \ty_i$ hold for all $i \leq n$.
  The typing environment $\Gamma;\Delta$ is divided into an \emph{exponential typing context} $\Gamma$ and an \emph{affine typing context} $\Delta$.  The exponential context deals with all the clonable data, including the type $\ctyp$ for classical basic data.
The affine context aims at ensuring that variables of type $\qtyp$, corresponding to quantum data, can be used at most once, hence forbidding their cloning.
Notice that this allows discarding variables of type $\qtyp$. Since the quantum data is given within the terms and there is no construction allowing a variable to refer to a part of a quantum state, a variable of type $\qtyp$ cannot be entangled with another state and can thus be safely discarded.
The outcome of a measurement is typed by the type $\Out\in\Qbasic$ to emphasize that measuring a quantum state results in a classical outcome together with a post-collapse quantum state (hence non-clonable).
The  rule $(\Rightarrow_e)$ for classical implication and the rule $(\mathsf{rec})$ for recursion  require the affine context to be empty to avoid cloning quantum data during the reduction.

  From now on, we will only consider well-typed terms.

\section{Cost-Structure Transformers}

We now define a higher-order expectation transformer, that can be viewed as \emph{continuation-passing style} (CPS) transformation, 
  mapping terms of the abstract quantum language to \emph{cost-structure terms}.
\subsection{Cost-Structure Language}
As a first step, we introduce the \emph{cost-structure language}, an intermediate language designed for
the expectation-based quantitative analysis of quantum programs.
This intermediate language serves as a symbolic framework for reasoning about quantum programs. This viewpoint justifies including real numbers and permits quantum data duplication.

The set $\TTerms$ of \emph{cost-structure terms} and the set $\TValues$ of \emph{cost-structure values} are defined by the grammar of \Cref{fig:cost:Grammar}, where $\tx, \tf $ belong to a fixed countable set $\VAR$ of variables. 
Again, $\mathsf{fv}(\te)$ denotes the free variables of $\te$.
As terms of the cost-structure language are in essence results of a CPS translation, redexes will occur only in head position.  This is reflected in the syntax and can be viewed as a kind of  A-normal form~\cite{SF92} where only trivial arguments serve as operands of function application.

The cost-structure language still includes most of the constructs of the abstract language of Figure~\ref{fig:Grammar}: the standard constructs of the lambda calculus, the constructs for quantum data,  and the constructs for basic data (there is no more a distinction between quantum and classical data using a semi-colon) and recursion.
However, a new construct $\Meas{b}$ ($b \in \{0,1\})$ for the post-collapse quantum state and  \emph{cost-structure operators} (\emph{CS operators}, for short) have been substituted to the measurement construct and the tick construct. The  CS operators  build terms that represent cost-structure operations (see Section~\ref{ss:k}), but do not represent quantum programs: the (non-negative) real constants $\breal{r}$ ($r \in \Rpos$), $\cadd$ for adding costs in a cost-structure, $\cbary{{\cdot}}$ for computing barycentric sums in a Kegelspitze.

\begin{figure}[t]
  \centering
  \begin{bnf}
    \TTerms \ni  \te 
    & \tx | \lambda \tx. \te | \te \app \tv & $\lambda$-calculus \\
    & \ket\phi | \unary{\tv} | {\tv \qtimes \tv} | \Meas{b}{\tv}  & Quantum data \\
    & \cons(\vec{\tv}) & Constructor \\
    &  \caseof{\tv}[\cons(\vec{\tx})]{\te}[\txtwo]{\te} & Destr.\\
    & \letrec{\tf}{\tx}{\te} & Recursion \\
    & \breal{r}  | \te \cadd \te | \te \cbary{{\tv}} \te
    & CS operators\\[0.2cm]
    \TValues \ni \tv
    & \tx | \lambda \tx.\te | \ket{\psi} | \unary{\tv} | {\tv \qtimes \tv} | \Meas{b}{\tv} | \cons(\vec{\tv})  \\
    & \letrec{\tf}{\tx}{\te} | \breal{r}
  \end{bnf}
  \caption{Grammar for the Cost-Structure Language.}
  \label{fig:cost:Grammar}
\end{figure}

\begin{figure*}[!t]
  \begin{center}
    \begin{minipage}{.8\textwidth}
      \begin{bnf}
        \TTypes  \ni \tty  &  \typ | \realty | \kty | \tty \Rightarrow \tty & \TYPES \\
        \TFTypes \ni \fty &   \kty | \tty \Rightarrow \fty & \FTYPES \\
        \tctx & \emptyset | \tctx, \tx : \tty & CS contexts \\
      \end{bnf}
    \end{minipage}
  \end{center}
  \vspace{2mm}

  \begin{infers}
    \Infer[tty][ax]{\tctx, \tx:\tty \INF \tx:\tty}{\strut}
    &
    \Infer[tty][expi][\Rightarrow_i]{\tctx\INF\lambda \tx.\te: \tty\Rightarrow \tty'}{\tctx,\tx:\tty \INF \te:\tty'}
    &
    \Infer[tty][expe][\Rightarrow_e]{
      \tctx\INF \te \app \romanosmallbox{\tv}:\tty'
    }{
      \tctx\INF \te:\tty \Rightarrow \tty' & \tctx\INF \romanosmallbox{\tv}:\tty
    }
    &
    \Infer[tty][st]{\tctx\INF\ket{\psi}:\qbty}{\strut}
    \\
    \Infer[tty][un]{\tctx\INF \unary{\tv}:\qbty}{\tctx\INF \tv:\qbty}
    &
    \romanobox{ \Infer[cost][meas]{\tctx \INF \Meas{b}{\tv} : \qbty}{\tctx \INF \tv : \qbty} }
    &
    \Infer[tty][prod][\otimes]{\tctx\INF \tv_0\qtimes \tv_1:\qbty}{\tctx\INF \tv_0:\qbty & \tctx\INF \tv_1:\qbty}
    &
    \Infer[tty][cons]{
      \tctx  \INF \cons(\romanosmallbox{$\vec \tv$}):\typ
    }{
      \tctx \INF \romanosmallbox{$\vec \tv$}: \vec \typ  &\cons :: \vec \typ  \to \typ
    }
    \\
    \Infer[tty][case]{
      \tctx  \INF\caseof{\romanosmallbox{\tv}}[\cons(\vec\tx)]{\te_0}[\txtwo]{\te_1}:\tty
    }{
      \tctx \INF \romanosmallbox{\tv}\!:\!\typ
      & \tctx ,\vec \tx  \!:\! \vec \typ  \INF \te_0\!:\!\tty
      & \tctx ,\txtwo\!:\!\typ \INF \te_1\!:\!\tty
      &\cons \!::\! \vec \typ  \to \typ
    }
    &
    \Infer[cost][rec]{
      \tctx \INF \letrec{\tf}{\tx}{\te} : \romanosmallbox{$\fty$}
    }{
      \tctx, \tf: \romanosmallbox{$\fty$}\INF \lambda\tx.\te : \romanosmallbox{$\fty$}
    }
    \\
    \romanobox{
      \Infer[cost][real]{
        \tctx \INF \breal{r} : \realty
      }{
        r \in \Rpos
      }
    }
    &
    \romanobox{
      \Infer[cost][opplus][\cadd]{\tctx  \INF \te_0 \cadd \te_1 : \kty
      }{
        \tctx \INF \te_0 : \realty & \tctx \INF \te_1 : \kty
      }
    }
    &
    \romanobox{
      \Infer[cost][bary][\mathbin{\pmb{+}_{p_0}}]{
        \tctx  \INF \te_0 \cbary{{\tv}} \te_1 : \kty
      }{
        \tctx \INF \te_0 : \kty& \tctx \INF \te_1 : \kty & \tctx \INF \tv : \qbty
      }
    }
  \end{infers}
  \caption{Types and Rules for the Cost-Structure Language.}
  \label{fig:cost:ty}
\end{figure*}

\subsection{Type System}
Cost-structure terms follow a simple typing regime, see \Cref{fig:cost:ty}, where typing judgments have the usual form $\tctx \INF \te : \tty$ and are extended to sequence $\tctx \INF  \vec \te : \vec \tty$ in a standard manner (see~\Cref{ss:ts}). 
The type discipline differs from the one of the abstract quantum language (\Cref{fig:TS}) on the following points:
\begin{varitemize}
\item To model costs, we make use of a dedicated type $\kty$ to denote a generic cost-structure $\kost$ (Definition~\ref{def:cs}) and of the type $\realty$ to denote the non-negative real numbers or infinity.
\item \emph{Cost-structure types} ($\TYPES$, for short) $\tty \in \TTypes$ do not include linear arrows. In contrast to the abstract language, types of the cost-structure language neglect the separation between quantum and classical data. For example, the type $\qbty \Rightarrow \qbty$ can be derived. This is not an issue as the typing constraints imposed by the laws of quantum mechanics are already enforced on the abstract quantum language (\Cref{fig:TS}).
\item  Recursion is only permitted on terms of  \emph{Functional cost-structure type} (Functional CS Types, for short) $\fty \in \TFTypes$.  These are mostly terms whose return type is the CS Type $\kty$. Let $\Costfuns$ be the set of cost-structure terms of functional CS type. The  inclusion $\Costfuns \subsetneq \TTerms $ trivially holds.
\end{varitemize}
In \Cref{fig:cost:ty}, the new rules for cost-structure operators and main changes wrt. the type system of \Cref{fig:TS} are highlighted in \romanobox{gray}.

\subsection{Denotational Semantics}

\begin{figure*}[t]
  \columnwidth=\linewidth
\begin{minipage}{0.4\textwidth}
  \begin{align*}
    \ksem{\env}{\kost}{\tx} & \defsym \env(\tx) \\
    \ksem{\env}{\kost}{\lambda \tx. \te} & \defsym x \mapsto \ksem{\env\{\tx := x\}}{\kost}{\te} \\
    \ksem{\env}{\kost}{\te \app \tv} & \defsym \apply(\ksem{\env}{\kost}{\te},\ksem{\env}{\kost}{\tv})  \\
    \ksem{\env}{\kost}{\ket\psi} & \defsym \ket\psi \\
    \ksem{\env}{\kost}{\unary{\tv}} & \defsym \ext{U}\, \ksem{\env}{\kost}{\tv} \\
    \ksem{\env}{\kost}{\tv_0 \qtimes \tv_1} & \defsym \ksem{\env}{\kost}{\tv_0} \otimes \ksem{\env}{\kost}{\tv_1} \\
    \ksem{\env}{\kost}{\Meas{b}{\tv}} & \defsym  M_{ \ksem{\env}{\kost}{\pmb{b}}} \ksem{\env}{\kost}{\tv}
  \end{align*}
\end{minipage}
\begin{minipage}{0.58\textwidth}
  \begin{align*}
     \ksem{\env}{\kost}{\cons(\vec\tv)} & \defsym \cons(\ksem{\env}{\kost}{\vec\tv})\\
      {\ksem{\env}{\kost}{\caseof*{\tv}[\cons(\vec{\tx})]{\te_0}[\txtwo]{\te_1}}} & \defsym
    \begin{cases}
      \ksem{\env\{\vec{\tx}:=\ksem{\env}{\kost}{\vec{\tv}}\}}{\kost}{\te_0} & \text{if $\ksem{\env}{\kost}{\tv} = \cons(\ksem{\env}{\kost}{\vec{\tv}})$}\\
      \ksem{\env\{ \txtwo := {\ksem{\env}{\kost}{\tv}}\}}{\kost}{\te_1} & \text{otherwise}
    \end{cases}\\
    \ksem{\env}{\kost}{\letrec{\tf}{\tx}{\te}} & \defsym \lfp(f\mapsto (x \mapsto \ksem{\env\{\tf:= f;\tx :=  x\}}{\kost}{\te}))\\
      \ksem{\env}{\kost}{\breal{r}} & \defsym r , \quad r \in \Rpos \\
      \ksem{\env}{\kost}{\te_0 \cadd \te_1} & \defsym \ksem{\env}{\kost}{\te_0} \tplus \ksem{\env}{\kost}{\te_1}\\
    \ksem{\env}{\kost}{\te_0 \cbary{{\tv}} \te_1} & \defsym \ksem{\env}{\kost}{\te_0} \bary{{p_0 \ksem{\env}{\kost}{\tv}}} \ksem{\env}{\kost}{\te_1}
  \end{align*}
\end{minipage}
  \caption{Denotational Semantics of the Cost-Structure Language with $\kost = (\keg,\tplus)$.}
  \label{fig:denot}
\end{figure*}

We now endow the cost language
with a denotational semantics. 
For a fixed cost-structure $\kost$, types of the cost-structure language are interpreted (provided $\typ\neq\qbty$) by:
\begin{align*}
\typesem{\typ}  &\defsym \mparbox{1cm}{\bigcup_{\scriptsize \begin{array}{l}\cons ::  \vec \ctyp ; \vec\qtyp \to \typ\in \Cons(\typ)\\ (\vec V_\cbty,\vec  V_\qbt) \in \typesem{\vec \ctyp}\times \typesem{\vec \qtyp}\end{array}}\! \! \! \{ \cons( \vec V_\cbty, \vec  V_\qbt) \}} \\
\typesem{\qbty} & \defsym \cup_{n \in \N} \Comp^{2^n}\\
\typesem{\realty} & \defsym \Rext \\
\typesem{\kty}  &   \defsym \kost  \\
\typesem{\tty \Rightarrow \tty'}  &\defsym [\typesem{\tty} \tocont \typesem{\tty'}] 
\end{align*}
where $\typesem{\qbty}$ and $\typesem{\typ}$ are considered as discrete $\omega$-cpos (using identity as order).
The type $\realty$ is interpreted by $\Rext$ endowed with the \emph{vertical ordering}: $r_1 \sqsubseteq r_2$ iff $r_1 \leq_\R r_2$ or $r_2 = \infty$.

Recall that functions $f : \mathcal{D} \to \mathcal{E}$, with $\mathcal{D}$ a discrete cpo, are continuous;
composition of continuous functions is continuous $(\circ) : [\mathcal{E} \tocont \mathcal{F}] \to [\mathcal{D} \tocont \mathcal{E}]\to [\mathcal{D}\tocont \mathcal{F}]$; and application $\apply : [\mathcal{D} \tocont \mathcal{E}] \times \mathcal{D} \to \mathcal{E}$ is continuous. Whenever $\mathcal{D}$ has least element $\bot$,
the least fixed point $\lfp : [\mathcal{D} \tocont \mathcal{D}] \to \mathcal{D}$, defined by $\lfp(f) \defsym \lub_{n
  \in \N} f^n(\bot)$, where $f^n$ denotes the $n$-fold composition of $f$, is continuous. 
Hence, the interpretation of an abstraction is also continuous and the interpretation of any type is an $\omega$-cpo.

Let the \emph{valuation} $\rho$ be a partial map from the set of variables $\VAR$ to the set $\cup_{\tty \in  \TTypes} \typesem{\tty}$. 
A valuation $\rho$ is \emph{valid} for a CS context $\tctx$, denoted by $\rho \Vdash\tctx$, if for all $\tx : \tty \in \tctx$, we have $\rho(\tx) \in \typesem{\tty}$. 
Let $\emptyv$ be the defined nowhere valuation and $\rho\{\tx := \tv\}$ be the valuation extending $\rho$ by $\rho\{\tx := \tv\}(\tx)=\tv$. Let also $\rho\{\vec{\tx} := \vec{\tv}\}$ be the natural extension to sequences of the same length.

The denotational semantics of terms with respect to a cost-structure $\kost$ and a valuation $\rho$ is given in \Cref{fig:denot}.
The semantics is defined on typing judgments but to avoid cumbersome notations, it is presented only on terms.
That is, giving a typing judgment $ \tctx \INF \te : \tty$, the interpretation $\ksem{\env}{\kost}{\te}$ denotes an element of  $\typesem{\tty}$, for some  valuation $\rho$ such that $\rho \Vdash\tctx$ (see Theorem~\ref{thm:sound}).

In Figure~\ref{fig:denot}, $p_0 \ksem{\env}{\kost}{\tv}$ denotes the probability  of measuring the outcome $0$ on the first qubit of the quantum state $\ksem{\env}{\kost}{\tv}$ ($p_0$ is defined in Section~\ref{ss:os}) and $\ksem{\env}{\kost}{\Meas{b}{\tv}}$ denotes the collapse of the quantum state $\tv$ after a measurement with outcome $b \in \{0,1\}$. For example, for $\ket \psi = \sfrac{1}{\sqrt3} \ket{00} + \sfrac{1}{\sqrt3} \ket{10} + \sfrac{1}{\sqrt3} \ket{11}$, $p_1 \ket \psi$ evaluates to $\sfrac{2}{3}$ and $\ksem{\env}{\kost}{\Meas{1}{\ket \psi}}$ evaluates to  $\sfrac1{\sqrt{2}}\ket{10}+\sfrac{1}{\sqrt{2}}\ket{11}$ (the normalized state obtained after measuring the outcome $1$).

\begin{theorem}\label{thm:sound}
If $\tctx \INF \te : \tty$ then
$\ksem{\env}{\kost}{\te } \in \typesem{\tty}$ for all $\env \Vdash \tctx$.
\end{theorem}

\subsection{Quantum Expectation Transformers}

\begin{figure*}[t]
  \columnwidth=\linewidth
Mapping of Types $\ttpe{(\cdot)} : \Types \to \TTypes$
    \begin{align*}
      \ttpe{\typ} & \defsym \typ &
       \ttpe{(\ty_1 \multimap \ty_2)} & \defsym \ttpe{\ty_1} \Rightarrow (\ttpe{\ty_2} \Rightarrow \kty) \Rightarrow \kty
      &
       \ttpe{(\cty \Rightarrow \ty)}  &\defsym \ttpe{\cty} \Rightarrow (\ttpe{\ty} \Rightarrow \kty) \Rightarrow \kty
\\
    \end{align*}

Transformation of Values $\qetv{(\cdot)} : \Values \to \TValues$
    \begin{align*}
      \qetv{\sx} & \defsym  \tx &
      \qetv{(\lambda \sx. \st)} & \defsym \lambda \tx.\lambda \tk. \qet{\st}{\tk} &
      \qetv{\ket{\psi}} & \defsym \ket{\psi} \\
      \qetv{\cons(\vec\sv_0;\vec \sv_1)} & \defsym \cons(\qetv{\vec \sv_0},\qetv{\vec \sv_1}) &
      \qetv{(\letrec{\sf}{\sx}{\st})} & \defsym \letrec{\tf}{\tx}{\lambda \tk. \qet{\st}{\tk}} \\
    \end{align*}

Transformation of Terms $\qet{\cdot}{\cdot} : \Terms \to \Costfuns \to \TTerms$
    \begin{align*}
      \qet{\sv}{\cont} & \defsym \cont \app \qetv{\sv}
      \\
      \qet{\st_0 \app \st_1}{\cont} & \defsym \qet{\st_1}{\lambda \tx_1. \qet{\st_0}{\lambda \tx_0. \tx_0 \app \tx_1 \app \cont}}
      \\
      \qet{\unary{\st}}{\cont} & \defsym \qet{\st}{\lambda \tx. \cont \app \unary{\tx}}
      \\
      \qet{\meas{\st}}{\cont} & \defsym \qet{\st}{\lambda \tx. (\cont \app \inj_0(\Meas{0}{\tx}) \cbary{{\tx}} \cont \app \inj_1(\Meas{1}{\tx}))}
      \\
      \qet{\st_0 \qtimes \st_1}{\cont} & \defsym \qet{\st_1}{\lambda \tx_1. \qet{\st_0}{\lambda \tx_0. \cont\! \app\! (\tx_0 \qtimes \tx_1)}}
      \\
      \qet{\cons(\vec\st_0; \vec \st_1)}{\cont} & \defsym \qet{\vec \st_1}{\lambda \vec \tx_1. \qet{\vec \st_0}{\lambda \vec \tx_0. \cont \app \cons(\vec \tx_0; \vec \tx_1)}}
      \\
      \qet{\caseof*{\st}[\cons(\vec{\sx}_0;\vec{\sx}_1)]{\st_0}[\sy]{\st_1}
      }{\cont} & \defsym \qet{\st}{\lambda \tx. \caseof*{\tx}[{\cons(\vec{\tx_0},\vec{\tx_1})}]{\qet{\st_1}{\cont}}[{\txtwo}]{\qet{\st_2}{\cont}}}
      \\
      \qet{\tick{\st}}{\cont} & \defsym \bun \cadd \qet{\st}{\cont}
    \end{align*}
  \caption{Expectation Transformer.}
  \label{fig:ect}
\end{figure*}

Conceptually, our notion of higher-order quantum expectation transformer is an adaptation of the one-pass call-by-value CPS transformation by \cite{DanvyN03}. The transformation is depicted in \Cref{fig:ect}. The type translation $\ttpe{(\cdot)}$ maps types of the source quantum language to those of the cost-structure language. While this translation leaves (first-order) base types unchanged, function types are mapped to their CPS counterparts, which route return values through an explicit continuation.

On the pure, non-quantum fragment of our language, our transformation closely resembles that of Danvy and Nielsen. It is defined by two mutually recursive translations: the translation $\qetv{(\cdot)}$ on values and a translation $\qet{\cdot}{\cont}$ on terms, which is parameterized in a continuation $\cont$. The distinguishing feature of our setting is that ticks and quantum measurements are translated away to cost-structure expressions, with measurements abstracted by their expected outcomes. In this way, in $\qet{\st}{\cont}$, the parameter $\cont$ is to be understood as representing the continuation’s expected quantitative behavior.

To present the definition formally, we assume a one-to-one correspondence between source variables $\sx \in \Var$ and target variables $\tx \in \VAR$, using lower- and upper-case letters to distinguish them. To simplify notation, we extend the type and value translations homomorphically to sequences, and we extend $\qet{\cdot}{\cont}$ to sequences by stipulating
\begin{align*}
  \qet{{\varepsilon}}{\cont} &\defsym \cont \\  \qet{{\vec \st,\st }}{\cont} &\defsym \qet{\st}{\lambda \tx. \qet{{\vec \st\,}}{\cont \app \tx}},
\end{align*}
for $\varepsilon$ the empty sequence. Let us extend the type translation $\ttpe{(\cdot)}$ to homomorphically
to typing contexts $\cctx;\qctx$. The
translations preserve typing in the following way:
\begin{proposition}
  Let $\cctx;\qctx\vdash \st : \ty$. 
  \begin{enumerateenv}
    \item If $\st \in \Values$, then $\ttpe{\cctx};\ttpe{\qctx} \INF \qetv{\st} : \ttpe{\ty}$ ; and 
    \item If $\ttpe{\cctx};\ttpe{\qctx} \INF \cont : \ttpe{\ty} \to \kty$, 
    then $\ttpe{\cctx};\ttpe{\qctx} \INF \qet{\st}{\cont} : \kty$. 
  \end{enumerateenv}
\end{proposition}

\begin{example}[Quantum cointoss translation]\label{ex:cointoss:qet}
  Reconsider the illustrating quantum cointossing example from the introduction.
  The value translation $\qetv{\cointoss}$ yields a
  cost-structure term $\COINTOSS : \qbty \Rightarrow (\qbty \Rightarrow \kty) \Rightarrow \kty$.
  By unfolding definitions, we have:
  \begin{align*}
  \mparbox{3mm}{\COINTOSS \defsym
    \qetv{\left(\block[c]{\letrec{ct}{\sx}{
                \caseof**{\tick{(\meas \sx)}}
                    [\inj_0(;\sx_0)]{\sx_0}
    [\inj_1(;\sx_1)]{ct \app \unary[H]{\sx_1}}}}\right)}} \\
    & = \letrec{CT}{X, \cont}{
            \qet{
            \caseof**{\tick{(\meas \sx)}}
                    [\inj_0(;\sx_0)]{\sx_0}
                    [\inj_1(;\sx_1)]{ct \app \unary[H]{\sx_1}}
      }{\cont} }
      \\
    & = 
      \letrec{CT}{X, \cont}{
       \bun \cadd (\cont' \app \inj_0(\Meas{0}{X}) \cbary{X} \cont' \app \inj_1(\Meas{1}{X}))
      }
    \\[1mm]
    & \qquad \textit{where }
      \cont' = \lam{X}{
        \caseof**{\tx}
            [\inj_0(X_0)]{\cont \app X_0}
            [\inj_1(X_1)]{CT \app \unary[H]{X_1} \app \cont}
      }
    \\
  & \equiv
      \letrec{CT}{X, \cont}{
        \bun \cadd (\cont \app \Meas{0}{X} \cbary{X} CT \app \unary[H]{(\Meas{1}{X})} \app \cont)
      }
  \end{align*}
  In the last step, we have inlined the auxiliary continuation $\cont'$.
  Albeit not identical, the two terms are semantically equivalent. 
  Note how by inlining further the constant zero continuation $\lambda \tz.\ \bzero$, 
  this term is identical, up to renaming, to the cost function given in \Cref{ss:bev}.
\end{example}

Given a substitution $\sigma : \Var \to \cValues$ of the form $\sigma =[\sv_1 / x_1,\dots, \sv_n / x_n ]$, we use $\env_{\sigma}$ as a shorthand for the valuation $\{\tv_1 :=\ksem{}{\kost}{\qetv{\sv_1}}, \dots, \tv_n := \ksem{}{\kost}{\qetv{\sv_n}} \}$.
The following constitutes the technical soundness theorem.

\begin{theorem}\label{thm:qet-sound}
  Let $\kost = (\keg,\tplus)$ be a cost structure.
  For every term $\cctx;\qctx \vdash \st \ofty \typ$, 
  substitution $\sigma : \Var \to \cValues$ such that $FV(\st) \subseteq dom(\sigma)$, 
  and cost-structure term $\qetv{\cctx},\qetv{\qctx}\INF \te : \typ \to \kty$,
  we have
  \[
    \ksem{\env_{\sigma}}{\kost}{\qet{\st}{\te}}
    = \ecost[\disto]{\st\sigma} \tplus \evalue[\disto]{\st\sigma}{\ksem{\env_{\sigma}}{\kost}{\te}}. 
  \]
\end{theorem}
As corollary, we obtain a sound and complete methodology for reasoning about
expected costs and values on quantum programs.
In the following, we abbreviate the
  cost-structure $(\Rext,+)$ by 
  $\Rext$, and denote a cost-free cost-structure $(\keg,\fplus)$
  by $\kostf$.
\begin{corollary}\label{cor:qet}
  Let $\st$, $\sigma$ and $\te$ be as in \Cref{thm:qet-sound}. The following both hold:
  \begin{enumerateenv}
    \item $\ecost[\disto]{\st\sigma} = \ksem{\env_{\sigma}}{\Rext}{\qet{\st}{\lambda X.0}}$, and
    \item 
    $\evalue[\disto]{\st\sigma}{\ksem{\env_{\sigma}}{\kostf}{\te}} = \ksem{\env_{\sigma}}{\kostf}{\qet{\st}{\te}}$.
  \end{enumerateenv}
\end{corollary}

\newcommand{\prpred}{P}

\section{Analysis Through Refinement Types}
\Cref{cor:qet} shows how the kind of extensional properties that we
are interested in are expressible as cost-structure terms. This eliminates the need to reason
directly via quantum semantics, and enables classical analysis tools. To demonstrate this,
we adapt one such technique, a \emph{refinement type system} \cite{FP91}, to the cost-structure language.

\subsection{Refinement Types}

\begin{figure*}
  \begin{minipage}[t]{.8\textwidth}
    \begin{bnf}
      \Forms \ni \form
      & \prpred(\tv_1,\ldots,\tv_n) | \neg \form | \form\land\form | \form\lor \form | \form\Rightarrow \form| \forall \tz^\tty.\ \form | \exists \tz^\tty.\ \form
      & Formulae \\
      \RTypes  \ni \rty
      & \refine{\rtb}{\form} | (\tx \ofty \rty) \Rightarrow \rty | \forall \tx \ofty \rty.\ \rty
      & Refinement Types \\
      \rctx
      & \cdot | \rctx,\tx : \rty | \rctx, \form
      & Refinement Contexts
    \end{bnf}
  \end{minipage}

  \begin{spec}{Well-formedness of Types and Contexts}[\begin{array}{r}\INFwf \rctx\quad \rctx \INFwf \rty\end{array}]
    \begin{infers}
      \scalebox{1.0}{ \Infer{\phantom{\rty}\INFwf \cdot}{}}
      &
      \scalebox{1.0}{\Infer{ \INFwf \rctx,\tx : \rty }{\INFwf \rctx & \tx \notin \mathsf{dom}(\rctx) & \rctx \INFwf \rty}}
      &
      \scalebox{1.0}{\Infer{\INFwf \rctx, \form}{\INFwf \rctx& \skel{\rctx} \INFb \form : \Bool}}
      \\
      \scalebox{1.0}{\Infer{ \rctx \INFwf \refine{\rtb}{\form} }{
          \INFwf \rctx
          & \skel{\rctx},\tz : \rtb  \INFb \form : \Bool
        }}
      &
        \scalebox{1.0}{\Infer{ \rctx \INFwf (\tx \ofty \rty) \Rightarrow \rty' }{
          \INFwf \rctx
          & \rctx \INFwf \rty
          & \rctx,\tx : \rty \INFwf \rty'
        }}
      &
        \scalebox{1.0}{ \Infer{ \rctx \INFwf \forall \tx \ofty \rty.\ \rty' }{
          \INFwf \rctx
          & \rctx \INFwf \rty
          & \rctx,\tx : \rty \INFwf \rty'
        }}
    \end{infers}
  \end{spec}
  
  \begin{spec}{Typing Relation}[\rctx \INFr \te : \rty]
    \begin{infers}
      \scalebox{1.0}{\Infer[rty][ax][ax]{\rctx  \INFr \tx:\rty}{ \tx:\rty \in \rctx }}
      &
      \scalebox{1.0}{\Infer[rty][expi][\Rightarrow_i]{\rctx,\rctx' \INFr \lambda \tx.\te: (\tx \ofty \rty) \Rightarrow \rty'}{
          \rctx,\tx:\rty,\rctx' \INFr \te:\rty'
          & \tx\notin \mathsf{fv}(\rctx')
        }}
      &
      \scalebox{1.0}{\Infer[rty][expe][\Rightarrow_e]{\rctx \INFr \te \app \tv:\rty'[\tv / \tx]}{
          \rctx \INFr \te:(\tx \ofty \rty) \Rightarrow \rty'
          & \rctx \INFr \tv:\rty
        }}
      &
       \scalebox{1.0}{ \Infer[rty][st]{\rctx \INFr\ket{\psi}:\refine{\qbty}{\tz = \ket\psi}}{\strut}}
      \\
      \scalebox{1.0}{\Infer[rty][un]{
          \rctx\INFr \unary{\tv}:\refine{\qbty}{\tz = \unary{\tv} \land \form[\tv]}
        }{
          \rctx\INFr \tv:\refine{\qbty}{\form}
        }}
	&
	\scalebox{1.0}{\Infer[rty][prod][\otimes]{
          \rctx\INFr \tv_0\qtimes \tv_1:\refine{\qbty}{\tz = \tv_0 \qtimes \tv_1 \land \form_0[ \tv_0]  \land \form_1[\tv_1]}
        }{
          \rctx\INFr \tv_0:\refine{\qbty}{\form_0}
         & \rctx\INFr \tv_1:\refine{\qbty}{\form_1}
        }}
      \\
      \scalebox{1.0}{\Infer[rty][meas]{
          \tctx \INFr \Meas{b}{\tv} : \refine{\qbty}{\tz = \Meas{b}{\tv} \land \form[\tv]}
        }{
          \tctx \INFr \tv : \refine{\qbty}{\form}
        }}
      &
      \scalebox{1.0}{\Infer[rty][cons]{
          \rctx  \INFr \cons(\vec{\tv}):
          \refine[\tz]{\typ}{ \tz = \cons(\vec{\tv}) \land \overrightarrow{\land\form}[\vec{\tv}/\vec{\tz}]}
        }{
          \cons :: \vec{\typ} \to \typ
          & \rctx \INFr \overrightarrow{\tv : \refine{\typ}{\form}}
        }}
      \\
      \scalebox{0.92}{\Infer[rty][case]
        {\rctx  \INFr\caseof{\tv}[\cons(\vec{\tx})]{\te_0}[\txtwo]{\te_1}: \rty}
        {
          \cons :: \vec \typ  \to \typ
          & \rctx \INFr \tv : \refine{\typ}{\form}
          & \rctx, \txtwo : \refine{\typ}{\form}, \txtwo = \tv, \forall \vec{X}^{\vec{\typ}}.\ \txtwo \neq \cons(\vec{\tx}) \INFr \te_1 : \rty
          & \rctx, \overrightarrow{\tx : \typ}, \tv = \cons(\vec{\tx}) \INFr \te_0 : \rty
        }}
      \\
      \scalebox{1.0}{\Infer[rty][rec]{
          \rctx,\rctx' \INFr \letrec{\tf}{\tx}{\te} : \rrfty
        }{
            \rctx, \tf:  \rrfty,\rctx' \INFr  \lambda \tx .\te : \rrfty
            & \text{Admis}(\rctx \INFwf  \rrfty)
            & \tf \notin \mathsf{fv}(\rctx')
        }}
      &&
      \scalebox{1.0}{\Infer[rty][real]{\rctx \INFr \breal{r} : \refine{\realty}{\tz = \breal{r}}}{
          r \in \Rpos
        }}
      \\
      \scalebox{1.0}{\Infer[rty][opplus][\cadd]{
          \rctx  \INFr \te_0 \cadd \te_1 : \refine{\kty}{\exists \tz_0^{\realty}, \tz_1^{\kty}.\ \tz = \tz_0 \cadd \tz_1 \land \form_0[\tz_0] \land \form_1[\tz_1]}
        }{
          \rctx \INFr \te_0 : \refine{\realty}{\form_0} & \rctx \INFr \te_1 : \refine{\kty}{\form_1}
        }}
      \\
      \scalebox{1.0}{\Infer[rty][bary][\mathbin{\pmb{+}_{p_0}}]{
          \rctx \INFr \te_0\cbary{\tv} \te_1  : \{ \tz : \kty \mid {\exists \tz_0^{\kty}, \tz_1^{\kty}.\ \tz = \tz_0 \cbary{\tv} \tz_1 \land \form_0[\tz_0 ] \land\form_1[\tz_1 ]\land \form_2[\tv ]}\}
        }{
          \rctx \INFr \te_0 : \refine{\kty}{\form_0} & \rctx \INFr \te_1 : \refine{\kty}{\form_1}
          & \rctx \INFr \tv : \refine{\qbty}{\form_2}
        }}
      \\
      \scalebox{1.00}{\Infer[rty][Gen][\forall_i]{ \rctx, \rctx' \INFr \te : \forall \tx : \rty'.\ \rty }{
          \rctx, \tx : \rty', \rctx' \INFr \te : \rty
          & \tx \notin \mathsf{fv}(\rctx'), \mathsf{fv}(\te)
        }}
      & 
      \scalebox{1.00}{\Infer[rty][Inst][\forall_e]{ \rctx \INFr \te : \rty[\tv / \tx] }{
          \rctx \INFr \te : \forall \tx{:}\rty'.\ \rty
          &
          \rctx \INFr \tv : \rty'
        }}
      &
      \scalebox{1.00}{\Infer[rty][RTsub][\rsubtype]{ \rctx \INFr \te : \rty }{
          \rctx \INFr \te : \rty' & \rctx \INFs \rty' \rsubtype \rty
        }} \!\!
    \end{infers}
  \end{spec}

  \begin{spec}{Subtyping Relation}[\rctx \INFs \rty \rsubtype \rty']
    \begin{infers}
      \scalebox{1.0}{\Infer[rsubtype][refl][re]{\rctx \INFs \rty \rsubtype \rty}{}}
      &
      \scalebox{1.0}{\Infer[rsubtype][trans][tr]{\rctx \INFs \rty_0 \rsubtype \rty_2}{
          \rctx \INFs \rty_0 \rsubtype \rty_1
          & \rctx \INFs \rty_1 \rsubtype \rty_2
        }}
      &
      \scalebox{1.0}{\Infer[rsubtype][basic][ba]{
          \rctx \INFs \refine[\tz]{\rtb}{\form} \rsubtype \refine[\tz]{\rtb}{\formtwo}
        }{
          \rctx \VALID \forall \tz^\rtb.\ \form \Rightarrow \formtwo
        }}
      \\
      \scalebox{1.0}{\Infer[rsubtype][arr][ar]{
          \rctx \INFs \tx : \rty_0 \Rightarrow \rty_0' \rsubtype \tx : \rty_1 \Rightarrow \rty_1'
        }{
          \rctx \INFs \rty_1 \rsubtype \rty_0
          & \rctx, \tx : \rty_1 \INFs \rty_0' \rsubtype \rty_1'
        }}
    \end{infers}
  \end{spec}

\caption{Refinement Types, Typing Rules, Well-Formedness and Subtyping}
\label{fig:refine:sub}
\end{figure*}

Refinement types can be seen as a specific form of dependent type system where
logical formulae are used to constrain base types. Our system, depicted in Figure~\ref{fig:refine:sub}
extends upon the simply typed system imposed on the cost-structure language in 3 essential ways:
\begin{varitemize}
  \item base types $\refine{\rtb}{\form}$ permit refining base types $\rtb \in \{\kty,\realty\} \cup \Basic$
  of $CS$ to values $\tz$ satisfying the formula $\form$;
  \item function types $(\tx \ofty \rty) \Rightarrow \rty'$ become dependent, where the term variable $\tx$ of type $\rty$ may occur in $\rty'$;
  \item universal quantification $\forall \tx\! \ofty\! \rty.\ \rty'$ allows the bounded quantification
  of $\tx$, constrained by refinement type $\rty$, within $\rty'$.
\end{varitemize}
For instance, $\Nat^{> 0} = \refine{\Nat}{\tz \neq 0}$ defines the type of strictly positive natural numbers,
or
\[
  (X : \Nat^{> 0}) \Rightarrow (Y : \Nat^{>0}) \Rightarrow \refine{\Nat}{Z \leq X \land Z \leq Y}
\]
refers to a function that takes two strictly positive naturals, and produces as output a natural smaller than
both inputs.

As indicated in Figure~\ref{fig:refine:sub},
formulae $\form$ are first-order formulae over Boolean predicates $P$ (e.g., $P(X_1,X_2,X_3) \defsym X_1 = X_2 \otimes X_3$,
stating that $X_1$ is the result of tensoring $X_2$ and $X_3$).
While we keep the set of predicates open, we do assume that formulae $\form$ follow a simple typing regime, and write
$\tctx \INFb \form : \Bool$ to say that $\form$ is a well-typed formula, with variables taking types
according to the context $\tctx$. This way, we ensure that refinements can always be interpreted as Booleans.
We write $\form[\tv / \tx]$ for the formula obtained by replacing all occurrences of $\tx$ by the cost-structure term $\te$ in $\form$, or simply $\form[\tv]$ when $\tx$ is clear from context. This notation extends to
refinement types. Note however that, in $\refine{\rtb}{\form}$, $(\tx \ofty \rty) \Rightarrow \rty'$
and $\forall \tx : \rty.\ \rty'$
the variables $\tz$ and $\tx$ are bound in $\form$ and in $\rty'$, respectively.
To avoid notational overhead, we write the unrefined type $\refine{\rtb}{\top}$ simply as $\rtb$, and
$(\tx \ofty \rty) \Rightarrow \rty'$ as $\rty \Rightarrow \rty'$ when $\tx$ does not occur free in $\rty'$.

Typing judgments of the refinement type system take the form
$
  \rctx \INFr \te : \rty
$
where $\rctx$ is a refinement context, $\te$ a term in the cost-structure language, and $\rty$ a refinement type. Context $\rctx$ are lists of variable bindings $\tx : \rty$ and formulae $\form$.
The symbol $\cdot$ is used to denote the empty context. In a context $\rctx, \tx : \rty, \rctx'$,
the variable binding $\tx : \rty$ binds the variable $\tx$ in $\rctx'$ (but not in $\rctx$).
The set of variables occurring free in $\rctx$ are denoted by $\mathsf{fv}(\rctx)$.
Formulae in typing contexts will be used to incorporate a path-based analysis into the typing regime,
they restrict the set of environments $\rho$ that a context will be interpreted as.
For instance, the context
\[
  X : \Nat, Y : \refine{\Nat}{Z \leq X}, X \neq Y
\]
refers to all environments $\rho$ where $\rho(X)$ is strictly larger than $\rho(Y)$.

Not every judgment is well-formed. For instance, in $\rctx \INFr \te : \rty$, we require that each variable in $\rctx$ is bound exactly once in a typing
context, and that $\rty$ is closed with respect to $\rctx$. Well-formedness of contexts and types (under a context)
is made precise in
Figure~\ref{fig:refine:sub} using the judgments $\INFwf \rctx\text{ and }\rctx \INFwf \rty$.
In the rules, $\skel{\rty}$ ($\skel{\rctx}$) refers to simple type (context) obtained by dropping all refinements:
\begin{align*}
 \skel{ \refine{\rtb}{\form}} &\triangleq \rtb \\
 \skel{(\tx \ofty \rty) \Rightarrow \rty'} &\triangleq \skel{\rty} \Rightarrow \skel{\rty'}\\
  \skel{\forall \tx \ofty \rty.\ \rty'} &\triangleq \skel{\rty'}
  \end{align*}
 We extend this definition to RT contexts in the expected way.
The well-formedness will be necessary to the well-definedness of the interpretation of typing judgments.
From now on, whenever we write a judgment $\rctx \INFr \te : \rty$, we implicitly assume it is well-formed,
i.e., $\rctx \INFwf \rty$.

\subsection{Typing Rules}
They are presented in Figure~\ref{fig:refine:sub}.
Rules~\ref{rty:ax}, \ref{rty:expi}, and \ref{rty:expe} are the usual rules for variables, $\lambda$-abstraction
and application. Note that in the latter, the value $V$ is substituted for the variable $X$ in the return-type $\rty'$.
Rules~\ref{rty:un}--\ref{rty:meas} directly reflect the semantics of the considered term within the type.
Rules~\ref{rty:cons} and \ref{rty:case} deal with constructors and destructors.
In the rules, the notation $\rctx \INFr \overrightarrow{\tv : \refine{\typ}{\form}}$
is shorthand for $\rctx \INFr \tv_i:  \refin{\typ_i}{\form_i}{\tz_i}$, for $1 \leq i \leq n$.
Similarly, $\overrightarrow{\land \form}$ is a shorthand notation for $\form_1 \land \ldots \land \form_n$ and $[\vec{\tv}/\vec{\tz}]$ is the natural extension of substitution to sequences.

Rule~\ref{rty:rec}, for recursive definitions, is the most interesting rule, because it deals
with the non-standard semantics given to recursive functions in the CS language. Soundness of the rule
hinges on the admissibility condition $\text{Admis}(\rctx \INFwf \rrfty)$, which ensures that the fixpoint generated by the recursive definition preserves the refinement type.
Admissibility allows inductive reasoning, whilst not imposing a termination condition.
This notion, defined semantically for brevity, will be formally set up in Definition~\ref{def:admissible}.

Rules~\ref{rty:real}, \ref{rty:opplus} and \ref{rty:bary} follow again the semantics of CS terms. In the latter two, the existentially quantified variables refer to the values that arguments of the corresponding operator evaluate to.
Rules~\ref{rty:Gen} and \ref{rty:Inst} permit generalization and instantiation.
While in dependently typed languages, universal type quantification is usually interpreted as a dependent function with the corresponding abstraction and type application, in our system,
universal quantification will be interpreted as an intersection. This way, we avoid explicit type applications in the language.
Finally, as customary for refinement type systems, Rule~\ref{rty:RTsub} introduces a subtyping relation into the type system.
Subtyping essentially permits weakening on refinements.

Formally, subtyping is stated as a judgment of the form $\rctx \INFs \rty \rsubtype \rty'$,
stating that $\rty$ is a subtype of $\rty'$ under RT context $\rctx$. Again, we tacitly assume well-formedness, i.e., $\rctx \INFwf \rty$ and $\rctx \INFwf \rty'$. The rules for subtyping
are given in \Cref{fig:refine:sub}.
All rules are standard. As usual, the case for comparing base types, Rule~\ref{rsubtype:basic},
hinges on a semantic condition. 
  Formally, given a formula $\form$ and refinement context $\rctx$, the judgment $\rctx \VALID \form$ asserts
  that every environment $\env$ adhering to the refinement context (i.e., $\env \VALID \rctx$) also satisfies the formula $\form$ (i.e., $\env \VALID \form$). We make these satisfaction relations, and the type semantics underlying them, precise below.

\subsection{Interpretation}
We endow the refinement types with an interpretation. 
To this end, we assume that each predicate $P$ comes with a fixed semantic $\rtsem{P}$, a relation in  $\typesem{\tty_1} \times \ldots \times \typesem{\tty_n}$, with $\tty_1, \ldots \tty_n \in \TTypes$.
  This way, we can define satisfaction of a formula $\form$ w.r.t. to an environment $\env$ 
  in a standard way, see \Cref{fig:semforms}. 
We can now interpret refinement types in the usual way.
This interpretation is extended to refinement types, and then to RT contexts, in the expected way.
An environment $\env$ adheres to a context $\rctx$, in notation $\env \VALID \rctx$, if $\env \in \rtsem{\rctx}$. 
Notice that well-formedness $\rctx \INFwf \rty$ allows us to interpret
the type $\rty$ as a subset of the interpretation of the underlying simple type $\skel{\rty}$, indexed
by valuations $\env \VALID \rctx$, i.e., $\rtsem[\env]{\rty} \subseteq \typesem{\skel{\rty}}$.

Using the interpretation of refinement types, it is now possible to formally define the notion of admissibility used in Figure~\ref{fig:refine:sub}.
\begin{definition}[Admissibility]\label{def:admissible}
  Given a cost-structure $\kost$, we say that $\mathbb{S} \subseteq \kost$ is an \emph{admissible subset} of $\kost$ if $\bot\in \mathbb{S}$ and, for any $\omega$-chain $x_0\sqsubseteq x_1 \sqsubseteq \cdots$ in $\mathbb{S}$, $\lub_{n\in\N}x_n \in \mathbb{S}$. I.e., $\mathbb{S}$ is a pointed $\omega$-cpo.

  The judgment $ \rctx \INFwf \rty $ is \emph{admissible}, noted  $\text{Admis}(\rctx \INFwf \rty)$,  if
  \begin{itemize}
    \item $\rty = \refine{\kty}{\form}$ and
      $\rtsem[\env]{\rty}$ is an admissible subset of $\kost$ for every $\env \VALID \rctx$.
    \item $\rty = (\tx \ofty \rty') \Rightarrow \rty''$ and $\text{Admis}(\rctx, \tx : \rty' \INFwf \rty'')$ holds.
    \item $\rty =  \forall \tx \ofty \rty'.\ \rty''$ and $\text{Admis}(\rctx, \tx : \rty' \INFwf \rty'')$ holds.
  \end{itemize}
\end{definition}
The refinement type $(\tx \ofty \rty) \Rightarrow \refine{\kty}{\tz \sqsubseteq e}$ will be admissible when $\tz$ does not appear free in $e$.
For example, the type $(\tx \ofty \rty) \Rightarrow \refine{\realty}{\tz \leq X^2+X+1}$ is admissible,
whereas $\refine{\kty}{1 \sqsubseteq \tz}$ is not (as its interpretation does not contain $\bot$).
Thus, Rule~\ref{rty:rec} allows to derive upper bounds, but not lower bounds.

\begin{figure}[t]
  \columnwidth=\linewidth
   \begin{spec}{Satisfaction of Formulae (excerpt)}
     \vspace*{-0.1cm}
     \begin{alignat*}{2}
       \env & \VALID \prpred(\te_1,\ldots,\te_n)&&\quad\Leftrightarrow\quad  (\ksem{\env}{\kost}{\te_1},\ldots,\ksem{\env}{\kost}{\te_n})\in\rtsem{\prpred}\\
       \env & \VALID \form_1\land\form_2&&\quad\Leftrightarrow\quad  \env \VALID \form_1 \mbox{ and } \env \VALID \form_1 \\
       \env & \VALID \exists \tz^{\tty}.\form &&\quad\Leftrightarrow\quad \mbox{ Exists $v \in\ksem{\env}{\kost}{\tty}$, }\env\{\tz :=v\} \VALID \form \\
       &  && \quad \ldots
     \end{alignat*}
   \end{spec}
  \vspace*{-0.1cm}
  \begin{spec}{Interpretation of Refinement Types}
     \vspace*{-0.1cm}
    \begin{align*}
      \rtsem[\env]{\refine{\rtb}{\form}}
      & \defsym \{ v \in \typesem{\rtb} \mid \rtsem[{\env\{\tz := v\}}]{\form} \}\\
      \rtsem[\env]{(\tx \ofty \rty) \Rightarrow \rty'}
      & \defsym
        \begin{lalign}
          \{\ f \in \typesem{\skel{\rty} \Rightarrow \skel{\rty'}} \\
          \mid \forall v \in \rtsem[\env]{\rty},\ \apply(f,v) \in \rtsem[\env\{\tx := v\}]{\rty'} \!\}
        \end{lalign} \\
      \rtsem[\env]{\forall \tx \ofty \rty.\ \rty'} & \defsym
              \begin{lalign}
      \{ v' \in \typesem{\skel{\rty'}}
      \!\mid  \!\forall v \in \rtsem[\env]{\rty}\!,\,v' \in \rtsem[\env\{\tx := v\}]{\rty'} \}
        \end{lalign}
    \end{align*}
  \end{spec}
   \vspace*{-0.1cm}
  \begin{spec}{Interpretation of Refinement Contexts}
    \vspace*{-0.1cm}
    \begin{align*}
      \rtsem{\cdot} & \defsym \{\emptyset\} \\
      \rtsem{\rctx,\tx : \rty} & \defsym \{ \env\{\tx := v\} \mid \env \in \rtsem{\rctx}, v \in \rtsem[\env]{\rty} \} \\
      \rtsem{\rctx,\form} & \defsym \{\rho \in \rtsem{\rctx} \mid \rtsem[\rho]{\form}\}
    \end{align*}
  \end{spec}
\caption{Interpretation of Formulae and Refinement Types.}
  \label{fig:semforms}
\end{figure}

\subsection{Soundness}
We now present the main technical result of this section, \emph{soundness of the
refinement type system}, informally, stating that typable terms adhere to their type specification.
Formally, the semantic interpretation of a typable term lies within the interpretation of its corresponding refinement type, as stated below.
\begin{theorem}[Type soundness]\label{thm:rts}
  If $\rctx \INFr \te : \rty$ then
  $\ksem{\env}{\kost}{\te} \in \rtsem[\env]{\rty}$
  for every $\env \VALID \rctx$.
\end{theorem}

  Coupled with soundness of the pre-expectation transformer (Corollary~\ref{cor:qet}) we obtain
  a sound methodology for reasoning about quantitative properties of quantum programs:

  \begin{corollary}\label{cor:qet-refine}
    Let $\cctx;\qctx \vdash \st \ofty \typ$ be a term.
    The following both hold:
    \begin{enumerateenv}
      \item If $\rctx \INFrr \qet{\st}{\lam{\tx}{0}} : \rty$ then
      $\ecost[\disto]{\st\sigma} \in \rtsem[{\env_\sigma}][\kost]{\rty}$
      for any substitution $\sigma$ with $\env_\sigma \VALID[\Rext] \rctx$,
      and 
      \item
        if $\rctx \INFr[\kostf] \qet{\st}{\te} : \rty$ then $\evalue[\disto]{\st\sigma}{\ksem{\env_{\sigma}}{\kostf}{\te}} \in \rtsem[{\env_\sigma}][\kostf]{\rty}$ for any
        cost-structure term $\qetv{\cctx},\qetv{\qctx}\INF \te : \typ \to \kty$
        and substitution $\sigma$ with $\env_\sigma \VALID[\Rext] \rctx$. 
     \end{enumerateenv}
  \end{corollary}

\section{Examples}
\label{sec:examples}

In this section we demonstrate our methodology on three examples.

\begin{example}[Quantum cointoss]\label{ss:cointoss}

We start by re-visiting the simple cointossing example from the introduction and show
how our methodology can be used to reason about its expected cost.
Reconsider the cost-structure function
\begin{align*}
  \COINTOSS \defsym {}
  & \letrec*{CT}{X, \cont}{
    \bun \cadd (\cont \app \Meas{0}{X} \cbary{X} CT \app \unary[H]{(\Meas{1}{X})} \app \cont)
    }
\end{align*}
computed in \Cref{ex:cointoss:qet}.
We will attribute $\COINTOSS$ the refinement type 
\[
  \rty \defsym (X:\qty) \Rightarrow \refine{\realty}{\tz \leq c(X)},
\]
where the  function $c$ is yet to be determined.
Let $\rctx \defsym CT:\rty,X:\qty$.

Proceeding inside-out, using
Rules~\ref{rty:ax}, \ref{rty:meas}, \ref{rty:un}, respectively,
and subtyping, one can derive
$
  \rctx
  \INFrr \unary[H]{(\Meas{1}{X})}
  : \qty
,
$
\noindent thus $$
  \rctx \INFrr CT \app \unary[H]{(\Meas{1}{X})} :
  \refine{\realty}{\tz \leq c(\unary[H]{(\Meas{1}{X})})}
  \,,$$
with Rules~\ref{rty:ax} and \ref{rty:expe}. Using Rules~\ref{rty:bary} and \ref{rty:real}
it follows:
\begin{multline*}
\rctx \INFrr \bzero \cbary{X} CT \app \unary[H]{(\Meas{1}{X})} :\\
  \refine{\realty}{\exists \tz_0, \tz_1.\,
  \tz = \tz_0 \cbary{X} \tz_1 \land \tz_0 = \bzero \land \tz_1 \leq c(\unary[H]{(\Meas{1}{X})})}
\end{multline*}
thus, by Rule~\ref{rty:RTsub},
\[
  \rctx
  \INFrr\!\!\!
  \bzero \cbary{X} CT \app \unary[H]{(\Meas{1}{X})} :
  \refine{\realty}{\tz \leq p_1(X) \times c(\unary[H]{(\Meas{1}{X})})}
  .
\]
Reasoning similarly using Rules~\ref{rty:opplus}
and \ref{rty:RTsub}, we derive
\begin{multline*}
  \rctx
  \INFrr
  \bun \cadd  (\bzero \cbary{X} CT \app \unary[H]{(\Meas{1}{X})}) :\\
  \refine{\realty}{\tz \leq 1 + p_1(X) \times c(\unary[H]{(\Meas{1}{X})})}
\end{multline*}
and thus via another application of Rule~\ref{rty:RTsub},
\[
  \rctx
  \INFrr
  \bun \cadd  (\bzero \cbary{X} CT \app \unary[H]{(\Meas{1}{X})}) :
  \refine{\realty}{\tz \leq c(X)}
  \,,
\]
provided
\begin{equation}
 \label{eq:ct:constr}
\tag{\dag}
  \rctx \VALID[\Rext] \forall Z. \tz \leq 1 + p_1(X) \times c(\unary[H]{(\Meas{1}{X})}) \Rightarrow \tz \leq c(X) \,.
\end{equation}
By admissibility of $\rty = (X:\qty) \Rightarrow \refine{\realty}{\tz \leq c(X)}$, this will permit us to finally derive
$$\INFrr \COINTOSS : (X:\qty) \Rightarrow \refine{\realty}{\tz \leq c(X)}$$ by Rules~\ref{rty:rec} and \ref{rty:expi}. Observe that the
constraint \eqref{eq:ct:constr} can be shown to hold taking
$c(X) \defsym 1 + 2 \times p_1(X)$, since then
\begin{align*}
  1 + p_1(X) \times c(\unary[H]{(\Meas{1}{X})})
  & = 1 + p_1(X) \times (1 + 2 \times p_1(\unary[H]{(\Meas{1}{X})}) \\
  & = 1 + p_1(X) \times (1 + 2 \times \sfrac{1}{2}) = c(X).
\end{align*}
In conclusion
\[
  \cdot \INFrr \COINTOSS : (X:\qty) \Rightarrow \refine{\realty}{\tz \leq 1 + 2 \times p_1(X)},
\]
and hence also $  Y : \qty \INFrr \COINTOSS \app Y : \{\tz \leq 1 + 2 \times p_1(Y)\}.$

Putting things together, Corollary~\ref{cor:qet-refine} certifies that
the expected cost of $\COINTOSS \app \ket{\psi}$ is bounded from above by $1 + 2 \times p_1(\ket{\psi})$.
\end{example}

\vspace*{-0.1cm}
\begin{example}[Higher-order quantum cointoss]\label{ss:qwalk}
\newcommand{\qwalk}{\ensuremath{\mathsf{qwalk}}}
\newcommand{\QWALK}{\ensuremath{\mathsf{QWALK}}}
To demonstrate that our methodology also encompasses a modular analysis of higher-order functions,
we consider a higher-order variation of $\cointoss$, $\qwalk :  \qbty \multimap  (\qbty \multimap \qbty)  \Rightarrow \qbty $, generalizing the Hadamard
gate applied at each iteration to an arbitrary quantum state transformation $f : \qbty \multimap \qbty$.
\begin{flalign*}
  & \qwalk \defsym
    \letrec{qw}{\sx\, f}{
    \caseof**{\tick{(\meas \sx)}}
    [\inj_0(;\sx_0)]{\sx_0}
    [\inj_1(;\sx_1)]{qw \app (f \app \sx_1) \app f}
    }
\end{flalign*}
Notice that $\cointoss \app \sy \equiv \qwalk \app \sy \app (\lambda X. \unary[H]{X})$.

Proceeding in a similar way to the previous example, we obtain
\[
  \qet{\qwalk \app \sy \app g}{\lambda Z.\ \bzero}
  \equiv \QWALK \app Y \app G ,
\]
for
$
  \QWALK \defsym
  \letrec*{QW}{X\,F}{
    \bun \cadd (\bzero \cbary{X} F \app (\Meas{1}{X}) \app (\lambda X'. QW \app X' \app F)).
  }
$

Notice how the recursive call defines the continuation of
$F : \qbty \Rightarrow (\qbty \Rightarrow \realty) \Rightarrow \realty$.
As in the previous example, the central step lies now in attributing
a suitable refinement type to the recursive function. Specifically,
one can prove
\[
 \cdot \INFrr \QWALK :
  \begin{lalign}
    \forall c : (\qbty \Rightarrow \realty). ( X : \qbty)\\
    \Rightarrow
    \begin{lalign}
      (\refine{\qbty}{\tz = \Meas{1}{X}} \\
      \Rightarrow (X' : \qbty \Rightarrow \refine{\realty}{\tz \leq c(X')}) \\
      \Rightarrow \refine{\realty}{1 + (0 \bary{p_0(X)} \tz) \leq c(X)})
    \end{lalign} \\
    \Rightarrow \refine{\realty}{\tz \leq c(X)}\,.
  \end{lalign}
\]
The above type makes use of bounded quantification, with $c: \qbty \Rightarrow \realty$
referring to the expected cost of an invocation of $\QWALK$. This generality is essential,
the overall cost of $\qwalk$ essentially depends on the state transformation $f$, represented
by $F$ in $\QWALK$. The refinement type associated to this argument reflects the minimal
constraint necessary to derive an overall cost of $c(X)$.

To recover the bound inferred for $$\cointoss \app y \equiv \qwalk \app y \app (\lambda X. \unary[H]{X}),$$ we specify the type of $\QWALK$
taking $c(X) \defsym 1 + 2 \times p_1(X)$, noting that
$\qetv{(\lambda X.\ \unary[H]{X})} = \lambda X.\lambda K.\ K \app (\unary[H]{X})$ can be typed as:
\[
  \cdot \INFrr \qetv{(\lambda X.\ \unary[H]{X})} :
  \begin{lalign}
      \refine{\qbty}{\tz = \Meas{1}{X}} \\
      \hspace{-4mm}\Rightarrow (X' : \qbty \Rightarrow \refine{\realty}{\tz \leq 1 + 2 \times p_1(X')}) \\
      \hspace{-4mm}\Rightarrow \refine{\realty}{1 + (0 \bary{p_0(X)} \tz) \leq 1 + 2 \times p_1(X)}
    \end{lalign}
\]
\end{example}

\vspace*{-0.1cm}
\begin{example}[Grover search]\label{ss:grover}
\newcommand{\grover}{\ensuremath{\mathsf{grover}_{f}}}
\newcommand{\GROVER}{\ensuremath{\mathsf{GROVER}_{f}}}
\cite{G96}'s quantum search algorithm finds a marked item in an
unsorted database of $2^n$ items with a quadratic speedup over classical search.
It relies on two key operators: $U_f$ and $G$. The function $f : \mathbb{N} \to
\mathbb{B}$ determines whether an item is marked or not. The oracle operator $U_f$ flips the sign of the amplitude of the
marked state (encoding the solution to the search problem), while the
Grover operator $G$ amplifies the probability of the marked state.
Repeated application of $U_f$ to
flip the sign and $G$ to amplify the amplitude of the marked state continues
until a measurement is performed, revealing the solution with high probability
in $O(\sqrt{2^n})$ iterations.
The oracle $U_f$ for the function $f$ is defined by
$
  U_f\ket{x}\ket{y} \defsym \ket{x}\ket{y \oplus f(x)}\textrm{, for }x\in\{0,1\}^n.
$ Thus, 
$$U_f\ket{x}\ket{-} = (-1)^{f(x)}\ket{x}\ket{-}.$$

Grovers algorithm, for fixed $f$ on a data-set of $n$ entries, is given by the term $\grover : \Nat \Rightarrow \qbty$ defined as:
\begin{flalign*}
  & \grover \defsym \letrec{grov}{m}{
    \caseof**{m}
     [\zero]{\ket{\psi_0} \ } 
     [\suc(m';\!)]{\unary[G]{(\unary[U_f]{(grov \app m'))}}}
  }
\end{flalign*}
where $
  G \triangleq 2\ket{\phi}\bra{\phi} - I
$ and $\ket{\psi_0}  \triangleq (H \ket{0})^{\otimes n} \otimes \ket{-}$.

We are interested in binding the error probability,
described by $\qet{\grover \app \ite}{\mathsf{err}}$ with the
term $\emptyset \INF \mathsf{err} : \qbty \Rightarrow [0,1]$
measuring the error probability $1 - |\alpha_b|^2$, given 
a quantum state $\sum_{b_i \in \{0,1\}^n} \alpha_i \ket{b_i}$,
Specifically, we seek an upper-bound
depending on the number of iterations $\ite$.
Observe that $$\qet{\grover \app \ite}{\mathsf{err}} \equiv \GROVER \app \ITE \app \mathsf{err},$$ where
\begin{flalign*}
& \GROVER \defsym
  \letrec*{GROV}{M\,K}{
     \caseof**{M}
        [\zero]{K \app \ket{\psi_0}}
        [\suc(M';\!)]{GROV \app M' \app (\lambda X.\ K \app \unary[G]{(\unary[U_f]{X})})}
  }
\end{flalign*}
Noting $\mathbb{P} \triangleq ([0,1],\fplus)$, we can ascribe $\GROVER$ the following type:
\[
  \cdot \INFr[\mathbb{P}] \GROVER :
  \begin{lalign}
    \forall k: \qbty \Rightarrow [0,1].\\
    \ITE : \Nat \\
    \Rightarrow (X : \qbty \Rightarrow \refine{[0,1]}{Z \leq k(X)}) \\
    \Rightarrow \refine{[0,1]}{Z \leq k((\mathsf{G}\,\mathsf{U_f})^{\ITE} \,\ket{\psi_0})}
  \end{lalign}
\]
Thus, in particular, 
\[
  \ITE : \Nat \INFr[] \GROVER \app \ITE \app \mathsf{err} :
  \{
  \begin{lalign}
    Z : [0,1] \mid  Z \leq \mathsf{err}((\mathsf{G}\,\mathsf{U_f})^{\ITE} \,\ket{\psi_0}) \} \,.
  \end{lalign}
\]
As $\ksem{\{\ITE := i\}}{\mathbb{P}}{\mathsf{err}((\mathsf{G}\,\mathsf{U_f})^{\ITE} \,\ket{\psi_0})}=
\cos^2((2i+1)\sin^{-1}(\frac{1}{\sqrt{2^n}}))$ holds,
the above typing thus binds the error probability of $\GROVER$ in terms
of the iteration count $i$.
\end{example}

\section{Conclusion}
We have developed an expectation-transformer analysis for a universal higher-order quantum programming language and shown that it can be instantiated as a refinement type system.
Future research will focus on sound implementations of these techniques as well as extending them to extended quantum languages with quantum control.

\bibliography{biblio}

\newpage

\appendix

\section{Extended PARS Preliminaries}

Here, we extend notions and notation of Section~\ref{s:pars}, as a preliminary to our soundness
proofs. Noteworthy, we will \emph{not} assume that the PARS $\cdot \rew{\cdot} \cdot \subseteq A \times \Rpos \times  \Distr(A)$ is deterministic.
In this setting, a PARS induces a (weighted) \emph{reduction relation} over \emph{multi-distributions}~\cite{AvanziniLY20}.
A multi-distribution $\mdist$
is a \emph{multiset} $\{ (a_i,p_i) \}_{i \in I}$ over pairs $(a,p)$, denoted $a^p$,
with \emph{support} $\supp(\mdist) = \{ a_i \mid i \in I\} \subseteq A$ and probabilities $0 \leq p_i \leq 1$ sum up to at most one,
i.e., $\sum_{i \in I} p_i \leq 1$.
Multi-distributions generalise (sub)distributions in the way multiset generalise sets. With $\mdist(a) \defsym \sum_{a^p \in \mdist} p$
we denote total probability attributed to $a$ in $\mdist$.
The set of multi-distributions over $A$ is denoted by $\MDistr(A)$.
By identifying distributions $\delta \in \Distr(A)$ with the set $\{ a^{\delta(a)} \}_{a \in \supp(\delta)}$ we have $\Distr(A) \subseteq \MDistr(A)$.
$\MDistr(A)$ is closed under scaling $p \cdot \{ a_i^{q_i} \}_{i \in I} \defsym \{ a_i^{p \cdot q_i} \}_{i \in I}$, and
convolution sums $\sum_{i \in I} p_i \cdot \mu_i \defsym \biguplus_{i \in I} p_i \cdot \mu_i$, for probabilities $\sum_{i \in I} p_i \leq 1$.

The reduction relation ${\cdot \drew{\cdot} \cdot} \subseteq \MDistr(A) \times \Rext \times \MDistr(A)$, is generated by the following rules:
\[
  \Infer[drew][T]{ \{ a^1 \} \drew{0} \{ a^1 \} }{ a \not\rew }
  \qquad
  \Infer[drew][R]{ \{a^1\} \drew{k} \delta }{ a \rew{k} \delta }
\]
\[
  \Infer[drew][C]{
    \sum_{i \in I} p_i \cdot \mdist_i \drew{\sum_i p_i \cdot k_i} \sum_{i \in I} p_i \cdot \mdist_i'
  }{
    \forall i \in I, \mdist_i \drew{k_i} \mdist_i'
  }
\]
When $\rew{}$ is deterministic, $\drew{}$ can be equally stated in terms of distributions,
as we have done in Section~\ref{s:pars}.
If $\mdist \drew{k} \mdisttwo$, then $\mdisttwo$ is obtained by rewriting all of the elements in
the support of $\mdist$. The weight $k$ signifies the (expected) cost of this reduction step. In a reduction sequence
\[
  \rho: \mdist = \mdist_0 \drew{k_0} \mdist_1 \drew{k_1} \mdist_2 \drew{k_2} \cdots
  ,
\]
the multi-distributions $\mdist_i$ give the $i$-step reducts of $\mdist$. The sum $\sum_{j=0}^{i-1} k_j$ yields
the expected cost to reach $\mdist_i$.
Let $\NF[\rew]{\mdist} : \Distr(A)$ denote the sub-distribution of normal forms in $\mdist$, i.e.,
$\NF[\rew]{\mdist}(a) = \mdist(a)$ if $a \not \rew$, and $\NF[\rew]{\mdist}(a) = 0$ otherwise.
We write $\mdist \drew{k}^\infty \beta$ if there is a reduction sequence converging to normal form distribution $\beta = \sup_{i \in \mathbb{N}} \NF[\rew]{\mdist_i}$ with cost $k = \sum_{i=0}^\infty k_i \in \Rext$.
Following \cite{Faggian22},
we denote by $\LDistr[\drew](\mdist) \defsym \{ \beta \mid \mdist \drew{c}^\infty \beta\}$ the set of all such normal form distributions of $\mdist$.
We denote by $\ecost[\drew](\mdist) \defsym \sup \{ k \mid \mdist \drew{k}^\infty \beta\} \in \Rext$ the expected reduction cost of $k$.

The \emph{reflexive closure} of a PARS $\rew$ is its least extension $\rew{}{=}$ such that $a \rew{0}{=} \{a^1\}$.
In contrast to $\mdist \drew{k} \mdisttwo$, in $\mdist \drew{k}{=} \mdisttwo$ only some of the elements in the support of $\mdist$ are rewritten.
By definition, ${\drew{k}{=}} \subseteq {\drew{k}{}}$.

Given two relations $\mathcal{R} \subseteq A \times B$ and $\mathcal{S} \subseteq B \times C$, we write $\mathcal{R} \circ \mathcal{S}$ the relation defined by $a \mathrel{(\mathcal{R} \circ \mathcal{S})} c$ iff $\exists b \in B$ such that $a \mathrel{\mathcal{R}} b$ and  $b \mathrel{\mathcal{S}} c$ both hold.

\begin{lemma}\label{l:nf-nondec}
  If $\mdist \drew{k}{=} \mdisttwo$ ($\mdist \drew{k} \mdisttwo$) then $\NF[\rew]{\mdist} \leq \NF[\rew]{\mdisttwo}$
\end{lemma}
\begin{proof}
  The claim can be proven by induction on the derivation of $\mdist \drew{k}{=} \mdisttwo$. Crucial,
  in Rule~\ref{drew:R}, if $a \rew{k}{=} \delta$ then eiter $a$ is not a normal form, or $\delta = \{a^1\}$
  and thus it is preserved in the reduction.
\end{proof}

We define \emph{multistep reduction relations}
$\mdist \drew{k}{*} \mdisttwo$ if $\mdist = \mdist_0 \drew{k_0}{=} \cdots \drew{k_{n-1}}{=} \mdist_n = \mdisttwo$ and $k = \sum_{i=0}^{n-1} k_i$,
similar, $\mdist \drew{k}{+} \mdisttwo$ if $\mdist \drew{k_1} \circ \drew{k_2}{=} \mdisttwo$ and $k = k_1 + k_2$.
By definition, ${\drew{k}{=}} \subseteq {\drew{k}{*}} \subseteq {\drew{k}{+}}$. All relations are closed
under convex combinations, the former by definition, for the latter two this fact is proven by the following lemma.
\begin{lemma}\label{l:closure-conv-closed}
  Suppose for all $i \in I$, $\mdist_i \drew{k_i}{\star} \mdisttwo_i$, where $\star \in \{{*},{+}\}$.
  Then $\sum_{i \in I} p_i \cdot \mdist_i \drew{\sum_{i \in I} p_i \cdot k_i}{\star} \sum_{i \in I} p_i \cdot \mdisttwo_i$ for probabilities $\sum_{i \in I} p_i \leq 1$.
\end{lemma}
\begin{proof}
  Consider first the case $\star = *$.
  Patching short reductions with $\mdistthree \drew{0}{*} \mdistthree$,
  we can write all $\mdist_i \drew{k_i}{*} \mdisttwo_i$ as
  \[
    \mdist_i
    = \mdistthree_{i,0}
    \drew{k_{i,0}}{=} \mdistthree_{i,1}
    \drew{k_{i,1}}{=} \cdots
    \drew{k_{i,n}}{=} \mdistthree_{i,n}
    = \mdisttwo_i,
  \]
  for a fixed $n \in \N$, where $\sum_{j=1}^n k_{i,j} = k_i$. As we have
  \[
    \sum_{i \in I} p_i \cdot \mdistthree_{i,j} \drew{\sum_{i \in I} p_i \cdot k_{i,j+1}}{=} \sum_{i \in I} p_i \cdot \mdistthree_{i,j+1},
  \]
  by definition, for all $0 \leq j < n$,
  using the identity
  $\sum_{i \in I} p_i \cdot k_i
  = \sum_{i \in I} p_i \cdot \sum_{j=1}^n k_{i,j}
  = \sum_{j=1}^n \sum_{i \in I} p_i \cdot k_{i,j}$,
  the lemma follows by reasoning inductively.

  Concerning $\star = +$, using ${\drew{k}{+}} = {\drew{k_1} \circ \drew{k_2}{=}}$ where $k = k_1 + k_2$,
  convex closure of $\drew{}{+}$ follows by that of $\drew{}$ and $\drew{}{=}$.
\end{proof}

By \Cref{l:nf-nondec}, the notation $\mdist \drew{k}{}^\infty \beta$ extends to $\drew{\cdot}{\star}$ ($\star \in \{{=},{*},{+}\}$). While by the above inclusions it is trivial to see that $\mdist \drew{k}{}^\infty \beta$
implies $\mdist \drew{k}{\star}^\infty \beta$, for instance for $\star = ({=})$ the converse does not hold.
For instance, $\mdist \drew{0}{=}^{\infty} \NF[\rew]{\mdist}$ for arbitrary $\mdist$, as witnessed by the infinite reduction
$\mdist \drew{0}{=} \mdist \drew{0}{=} \mdist \drew{0}{=} \cdots$. In contrast, $\mdist \drew{0}^{\infty} \NF[\rew]{\mdist}$ only holds when $\supp(\mdist)$ is in normal form.
In the following, we prove that $\drew{k}^\infty$ coincides with $\drew{k}{+}^\infty$. To this end, we
employ the following permutation lemma.

\begin{lemma}\label{l:drew-permute}
  If $\mdist \drew{k}{=} \mdisttwo \drew{l} \mdistthree$ then
  there exists $\mdisttwo'$ and $l_1 + l_2 = l$ such that
  (i)~$\mdist \drew{k+l_1} \mdisttwo' \drew{l_2}{=} \mdistthree$; and
  (i)~$\mdisttwo \drew{l_1}{=} \mdisttwo'$.
\end{lemma}
\begin{proof}
  Fix $\mdist \drew{k}{=} \mdisttwo \drew{l} \mdistthree$. This reduction can be written as
  \[
    \mdist = \mdist_1 \uplus \mdist_2 \drew{k}{=}
    \mdist_1 \uplus \mdisttwo_2 \drew{l}
    \mdistthree_1 \uplus \mdistthree_2
    = \mdistthree
    ,
  \]
  where
  (i)~$\mdist_2 \drew{k} \mdisttwo_2$,
  (ii)~$\mdist_1 \drew{l_1} \mdistthree_1$ and
  (iii)~$\mdisttwo_2 \drew{l_2} \mdistthree_2$.
  Eagerly using~(ii) results in a sequence
  \[
    \mdist = \mdist_1 \uplus \mdist_2 \drew{k+l_1}
    \mdistthree_1 \uplus \mdisttwo_2 \drew{l_2}{=}
    \mdistthree_1 \uplus \mdistthree_2
    = \mdistthree
    .
  \]
  Letting $\mdisttwo' = \mdistthree_1 \uplus \mdisttwo_2$ now proves the lemma.
\end{proof}

\begin{lemma}\label{l:drew+=>drew}
  $\mdist \drew{k}{+}^\infty \beta$ implies $\mdist \drew{k}^\infty \beta$.
\end{lemma}
\begin{proof}
  Consider an infinite reduction
  \begin{equation}
    \label{drewplus:red}
    \tag{\dag}
    \mdist = \mdist_0 \drew{k_0}{=}
    \mdist_1 \drew{k_1}{=}
    \mdist_2 \drew{k_2}{=}
    \cdots
    ,
  \end{equation}
  with \emph{infinitely many proper} steps $\mdist_{n_i} \drew{k_{n_i}} \mdist_{n_i+1}$ ($n_0 < n_1 < \cdots$),
  converging a normal form distribution $\beta$.
  We turn this into an infinite reduction over $\drew{\cdot}$ converging to $\beta$.
  As $\mdist \drew{k}{+}^\infty \beta$ is witnessed by such a reduction with $k = \sum_i k_i$ ---tacitly employing
  \Cref{l:nf-nondec} to equate the the limit normal form distribution induced by the sequence $(\mdist_{n_i})_{i \in \N})$ with $\beta$---
  this concludes the Lemma.

  To obtain the desired sequence,
  repeatedly apply \Cref{l:drew-permute} to permute the first proper step
  $\mdist_{n_0} \drew{n_0} \mdist_{n_0+1}$ to front of the reduction \eqref{drewplus:red},
  yields a sequence
  \begin{equation*}
    \mdist =
    \mdisttwo_0 \drew{l_0}{}
    \mdisttwo_1 \drew{m_1}{*}
    \mdist_{n_0 + 1} \drew{k_{n_0+1}}{=}
    \cdots
    ,
  \end{equation*}
  of equal cost ($l_0 + m_1 = \sum_{i=0}^{n_0} k_i$) with $k_0 \leq l_0$,
  and where $\mdist_1 \drew{\cdot}{=} \mdisttwo_1$.
  Repeating such permutations $\ell$ times in total, thereby permuting the first $\ell$ proper steps to the front,
  yields
  a sequence
  \begin{equation*}
    \mdist =
    \mdisttwo_0 \drew{l_0}{}
    \cdots \drew{l_{\ell - 1}}
    \mdisttwo_{\ell} \drew{m_\ell}{*}
    \mdist_{n_\ell + 1} \drew{k_{n_\ell+1}}{=}
    \cdots
    ,
  \end{equation*}
  of cost equal to that of \eqref{drewplus:red}, and where
  $\sum_{i=0}^{\ell - 1 } k_i \leq \sum_{i=0}^{\ell - 1} l_i$ and
  $\mdist_\ell \drew{\cdot}{=} \mdisttwo_\ell$.

  Repeating this construction ad infinitum yields finally a reduction
  \begin{equation*}
    \mdist =
    \mdisttwo_0 \drew{l_0}{}
    \mdisttwo_1 \drew{l_1}{}
    \mdisttwo_2 \drew{l_2}{}
    \cdots
    ,
  \end{equation*}
  approaching $\beta$,
  where for each $\ell \in \N$,
  \[
    \sum_{n=0}^\ell k_i \leq \sum_{n=0}^\ell l_i \leq \sum_{n=0}^\infty k_i
    ,
  \]
  and
  $\mdist_\ell \drew{\cdot}{=} \mdisttwo_\ell$.
  Notice how the former proves $\sum_{n=0}^\infty l_i = \sum_{n=0}^\infty k_i$,
  the latter proves $\NF[\rew]{\mdist_\ell} \leq \NF[\rew]{\mdisttwo_\ell} \leq \beta$ for all $\ell \in \N$,
  via \Cref{l:nf-nondec}.
  In total $\mdist \drew{\sum_i k_i}^\infty \beta$ as desired.\qed
 \end{proof}

 \begin{lemma}
   For all multidistributions $\mu$,
   (i)~$\ecost[\drew{}{+}](\mdist) = \ecost[\drew](\mdist)$ and
   (ii)~$\LDistr[\drew{}{+}](\mdist) = \LDistr[\drew](\mdist)$.
 \end{lemma}
 \begin{proof}
   Using Lemma~\ref{l:drew+=>drew}, and the fact that its inverse implication trivially holds since
   ${\drew{}} \subseteq {\drew{}{+}}$.
 \end{proof}

\section{Soundness of Quantum Expectation Transformer}
The following straightforward lemma states that the quantum expectation transformer is sound with respect to typing.
\begin{lemma}The following results hold.
  \begin{enumerateenv}
    \item If $\cctx;\qctx \vdash \sv : \ty$ then $\ttpe{\cctx},\ttpe{\qctx} \INF \qetv{\sv} : \ttpe{\ty}$.
    \item If $\cctx;\qctx \vdash \st : \ty$
    and $\ttpe{\cctx},\ttpe{\qctx} \INF \cont : \ttpe{\ty} \Rightarrow \kty$
    then $\ttpe{\cctx},\ttpe{\qctx} \INF \qet{\st}{\cont} : \kty$.
  \end{enumerateenv}
\end{lemma}
\begin{proof}
  \begin{enumerateenv}
  \item Straightforward induction on the structure of the typing derivation $\cctx;\qctx \vdash \sv : \ty$.
  \item Straightforward induction on the structure of the term $t$.
    \qedhere
  \end{enumerateenv}
\end{proof}

\section{Proof of Theorem~\ref{thm:rts} (Soundness of the Refinement Type System)}
In this section, we proof the Soundness of the refinement type system, that is, Theorem~\ref{thm:rts}.
We start by some technical lemmas.

\begin{lemma}[Term substitution]
  \label{lem:substCS}
  $\ksem{\env}{\kost}{\te[\tv/\tx]}=\ksem{\env\{\tx:=\ksem{\env}{\kost}{\tv}\}}{\kost}{\te}$.
\end{lemma}
\begin{proof}
  Straightforward induction on the structure of $\te$.
\end{proof}

\begin{lemma}[Formulae Substitution]
  \label{lem:fsubst}
    $$\rtsem[\env]{\form[\tv/\tx]}=\rtsem[\env\{\tx:=\ksem{\env}{\kost}{\tv}\}]{\form}.$$
\end{lemma}
\begin{proof}
  By induction on the structure of $\form$.
  \begin{itemize}
    \item Case $\rtsem[\env]{\prpred(\te_1,\ldots,\te_n)[\tv/\tx]}$: Using Lemma~\ref{lem:substCS}, we have
      \begin{align*}
	\rtsem[\env]{\prpred(\te_1,\ldots,\te_n)[\tv/\tx]}
	&=\rtsem[\env]{\prpred(\te_1[\tv/\tx],\ldots,\te_n[\tv/\tx])}\\
	&\Leftrightarrow (\ksem{\env}{\kost}{\te_1[\tv/\tx]},\ldots,\ksem{\env}{\kost}{\te_n[\tv/\tx]})\in\rtsem{\prpred}\\
	&\Leftrightarrow (\ksem{\env\{\tx:=\ksem{\env}{\kost}{\tv}\}}{\kost}{\te_1},\ldots,\ksem{\env\{\tx:=\ksem{\env}{\kost}{\tv}\}}{\kost}{\te_n})\in\rtsem{\prpred}\\
	&\Leftrightarrow \rtsem[\env\{\tx:=\ksem{\env}{\kost}{\tv}\}]{\prpred(\te_1,\ldots,\te_n)}
      \end{align*}
    \item All the other cases are straightforward.
      \qedhere
  \end{itemize}
\end{proof}

\begin{lemma}[Refinement Types Substitution]
  \label{lem:subst}
    $$\rtsem[\env]{\rty[\tv/\tx]}=\rtsem[\env\{\tx:=\ksem{\env}{\kost}{\tv}\}]{\rty}.$$
\end{lemma}
\begin{proof}
  By induction on the structure of $\rty$.
  \begin{itemize}
    \item Case $\rtsem[\env]{\refine{\rtb}{\form}[\tv/\tx]}$:
      Using Lemma~\ref{lem:fsubst}, we have
      \begin{align*}
	\rtsem[\env]{\refine{\rtb}{\form}[\tv/\tx]}
	&= \{ v \in \typesem{\rtb} \mid \rtsem[\env\{\tz := v\}]{\form[\tv/\tx]} \} \\
	&= \{ v \in \typesem{\rtb} \mid \rtsem[\env\{\tz := v,\tx:=\ksem{\env}{\kost}{\tv}\}]{\form} \} \\
	&= \rtsem[\env\{\tx:=\ksem{\env}{\kost}{\tv}\}]{\refine{\rtb}{\form}}
      \end{align*}
    \item Case $\rtsem[\env]{((\tx' \ofty \rty) \Rightarrow \rty')[\tv/\tx]}$:
      \begin{align*}
	\rtsem[\env]{((\tx' \ofty \rty) \Rightarrow \rty')[\tv/\tx]} 
	=& \{ f \in \typesem{\skel{\rty} \Rightarrow \skel{\rty'}}
    \\
    &\mid \forall v \in \rtsem[\env]{\rty[\tv/\tx]},\\
    &\ \ \apply(f,v) \in \rtsem[\env\{\tx := v\}]{\rty'[\tv/\tx]} \} \\
	=& \{ f \in \typesem{\skel{\rty} \Rightarrow \skel{\rty'}} 
    \\
    & \mid \forall v \in \rtsem[\env\{\tx:=\ksem{\env}{\kost}{\tv}\}\}]{\rty},\\
    &\ \ \apply(f,v) \in \rtsem[\env\{\tx := v,\tx:=\ksem{\env}{\kost}{\tv}\}]{\rty'} \}\\
	=& \rtsem[\env\{\tx:=\ksem{\env}{\kost}{\tv}\}]{(\tx \ofty \rty) \Rightarrow \rty'}
      \end{align*}
    \item Case $\rtsem[\env]{(\forall \tx' \ofty \rty.\ \rty')[\tv/\tx]}$:
      \begin{align*}
	\rtsem[\env]{(\forall \tx' \ofty \rty.\ \rty')[\tv/\tx]} 
	=& \{ f \in \typesem{\skel{\rty'}}\\
    & \mid \forall v \in \rtsem[\env]{\rty[\tv/\tx]},\\
    & \ f \in \rtsem[\env\{\tx' := v\}]{\rty'[\tv/\tx]} \} \\
	=& \{ f \in \typesem{\skel{\rty'}}\\
    & \mid \forall v \in \rtsem[\env\{\tx:=\ksem{\env}{\kost}{\tv}\}]{\rty},\\
    &\ \ f \in \rtsem[\env\{\tx' := v\}]{\rty'} \}\\
	=& \rtsem[\env\{\tx:=\ksem{\env}{\kost}{\tv}\}]{\forall \tx' \ofty \rty.\ \rty'}
	\tag*{\qedhere}
      \end{align*}
  \end{itemize}
\end{proof}

\begin{lemma}[Subtyping soundness]\label{lem:subtypingSoundness}
  If $\rctx\INFs \rty \rsubtype \rty'$ then 
  $\rtsem[\env]{\rty} \subseteq \rtsem[\env]{\rty'}$ for every $\env \in \rtsem{\rctx}$.
\end{lemma}
\begin{proof}
  We proceed by induction on the size of the derivation.
  \begin{itemize}
      \item Rule \ref{rsubtype:refl}. This case is trivial by the reflexivity
	of the subset relation.

      \item Rule \ref{rsubtype:trans}. This case is trivial by the transitivity
	of the subset relation.

      \item Rule \ref{rsubtype:basic}, we have $\rctx\INFs
	\refine{\rtb}{\form} \rsubtype \refine{\rtb}{\formtwo}$ as a
	consequence of $\rctx\vDash \forall \tz^\rtb,\ \form
	\Rightarrow \formtwo$. Direct from the definition of $\vDash_\kost$.

      \item Rule \ref{rsubtype:arr}, we have $\rctx\INFs
	\tx:\rty_0\Rightarrow\rty_0' \rsubtype \rty_1\Rightarrow\rty_1'$ as a
	consequence of $\rctx\INFs \rty_1 \rsubtype \rty_0$ and
	$\rctx\INFs \rty_0' \rsubtype \rty_1'$.  We need to prove that
	$\forall\env\in\rtsem{\rctx}$,
	$\rtsem[\env]{\tx:\rty_0\Rightarrow\rty_0'}\subseteq\rtsem[\env]{\tx:\rty_1\Rightarrow\rty_1'}$.
	That is, $\forall f\in\typesem{\skel{\rty_0}\Rightarrow\skel{\rty_0'}}$
	such that $\forall v_0\in\rtsem[\env]{\rty_0}$,
	$\apply(f,v_0)\in\rtsem[\env\{\tx:=v_0\}]{\rty_0'}$, implies that
	$\forall v_1\in\rtsem[\env]{\rty_1}$,
	$\apply(f,v_1)\in\rtsem[\env\{\tx:=v_1\}]{\rty_1'}$.

	By the induction hypothesis,
	$\rtsem[\env]{\rty_1}\subseteq\rtsem[\env]{\rty_0}$, and thus
	$\apply(f,v_1)\in\rtsem[\env\{\tx:=v_1\}]{\rty_0'}$, and since also by
	the induction hypothesis,
	$\rtsem[\env]{\rty_0'}\subseteq\rtsem[\env]{\rty_1'}$, we are done.
	\qedhere
  \end{itemize}
\end{proof}
\begin{proof}[Proof of Theorem~\ref{thm:rts} (Type soundness)]
  Theorem statement:
  \emph{If $\rctx \INFr \te : \rty$ then $\ksem{\env}{\kost}{\te} \in \rtsem[\env]{\rty}$ for every $\env \in \rtsem{\rctx}$.}

  Proof:
  We proceed by induction on the size of the derivation.
    \begin{itemize}
      \item Rule \ref{rty:ax}, we have $\rctx\INFr  \tx : \rty$ with
	$\tx:\rty\in\rctx$.  We need to prove that $\forall \env \in
	\rtsem[\env]{\rctx},\ \ksem{\env}{\kost}{\tx} \in \rtsem[\env]{\rty}$.
	That is, $\forall \env \in \{\env'\{\tx:=v\}\mid
	\env'\in\rtsem{\rctx\setminus\{\tx:\rty\}}, v\in\rtsem[\env']{\rty}\}$,
	we have $\env(\tx) \in \rtsem[\env]{\rty}$, which is trivial.

      \item Rule \ref{rty:expi}, we have $\rctx,\rctx'\INFr \lambda \tx.\ \te :
	(\tx \ofty \rty) \Rightarrow \rty'$ as a consequence of $\rctx, \tx :
	\rty,\rctx'\INFr \te : \rty'$, with $X\notin\mathsf{fv}(\rctx')$. 
	We need to prove that $\forall \env \in
	\rtsem{\rctx,\rctx'},\ \ksem{\env}{\kost}{\lambda \tx.\ \te} \in
	\rtsem[\env]{(\tx \ofty \rty) \Rightarrow \rty'}$.  That is, $\forall
	\env \in \rtsem{\rctx,\rctx'},\ \ksem{\env}{\kost}{\lambda \tx.\ \te}
	\in [\typesem{\skel{\rty}} \tocont \typesem{\skel{\rty'}}]$,
	or, in an equivalent manner, $\forall v \in \rtsem[\env]{\rty},\
	\ksem{\env\{X:=v\}}{\kost}{\te} \in \rtsem[\env\{\tx := v\}]{\rty'}$,
	which is true by the induction hypothesis.

      \item Rule \ref{rty:expe}, we have $\rctx\INFr \te \app \tv :
	\rty'[\tv/\tx]$ as a consequence of $\rctx\INFr \te : (\tx \ofty
	\rty) \Rightarrow \rty'$ and $\rctx\INFr \tv : \rty$.  We need
	to prove that $\forall \env \in \rtsem{\rctx},\
	\ksem{\env}{\kost}{\te \app \tv} \in \rtsem[\env]{\rty'[\tv/\tx]}$, or
	equivalently, $\apply(\ksem{\env}{\kost}{\te},\ksem{\env}{\kost}{\tv})
	\in \rtsem[\env]{\rty'[\tv/\tx]}$.

	By the induction hypothesis $\ksem{\env}{\kost}{\te} \in
	\rtsem[\env]{(\tx \ofty \rty) \Rightarrow \rty'}$ and
	$\ksem{\env}{\kost}{\tv} \in \rtsem[\env]{\rty}$.  Since, by
	Lemma~\ref{lem:subst}, we have
	$\rtsem[\env]{\rty'[\tv/\tx]}=\rtsem[\env\{\tx:=\ksem{\env}{\kost}{\tv}\}]{\rty'}$,
	we are done.

      \item Rule \ref{rty:st}, we have
	$\rctx\INFr\ket{\psi}:\refine{\qbty}{\tz = \ket\psi}$.  We need
	to prove that $\forall \env \in \rtsem{\rctx},\
	\ksem{\env}{\kost}{\ket{\psi}} \in \rtsem[\env]{\refine{\qbty}{\tz =
	\ket\psi}}$, or equivalently, $\ket{\psi} \in \{
	v\in\typesem{\qbty}\mid\rtsem[\env\{\tz:=v\}]{\tz = \ket\psi} \}
	=\{\ket\psi \}$, which is trivial.

      \item Rule \ref{rty:un}, we have $\rctx\INFr
	\unary{\tv}:\refine{\qbty}{\tz = \unary{\tv} \land \form[\tv]}$ as a
	consequence of $\rctx\INFr \tv:\refine{\qbty}{\form}$.  We
	need to prove that $\forall \env \in \rtsem{\rctx},\
	\ksem{\env}{\kost}{\unary{\tv}} \in \rtsem[\env]{\refine{\qbty}{\tz =
	\unary{\tv} \land \form[\tv]}}$.  Since, by the induction
	hypothesis, $\form[{\ksem{\env}{\kost}{\tv}}]$ holds, we are done.

      \item Rule \ref{rty:prod}, we have $\rctx\INFr \tv_0\qtimes
	\tv_1:\refine{\qbty}{\tz = \tv_0 \qtimes \tv_1 \land \form_0[ \tv_0]
	\land \form_1[\tv_1]}$ as a consequence of $\rctx\INFr
	\tv_0:\refine{\qbty}{\form_0}$ and $\rctx\INFr
	\tv_1:\refine{\qbty}{\form_1}$.

	We need to prove that $\forall \env \in \rtsem{\rctx}$, 
	\begin{align*}
        \ksem{\env}{\kost}{\tv_0\qtimes \tv_1} \in
	    \rtsem[\env]{\refine{\qbty}{\tz = \tv_0 \qtimes \tv_1 \land \form_0[
	    \tv_0]  \land \form_1[\tv_1]}}.
    \end{align*}
  Since, by the induction
	hypothesis, we have that $\form_0[\ksem{\env}{\kost}{\tv_0}]$ and
	$\form_1[\ksem{\env}{\kost}{\tv_1}]$ hold, we are done.

      \item Rule \ref{rty:cons}, we have $\rctx  \INFr
	\cons(\vec{\tv}):\refine[\tz]{\typ}{ \tz = \cons(\vec{\tz}) \land
	\overrightarrow{\land \form}}[\vec{\tv}/\vec{\tz}]$ as a consequence of
	$\rctx \INFr \overrightarrow{\tv} :
	\overrightarrow{\refine{\typ}{\form}}$ with $\cons :: \vec{\typ} \to
	\typ$.  We need to show that $\forall \env \in \rtsem{\rctx},\
	\ksem{\env}{\kost}{\cons(\vec{\tv})} \in
	\rtsem[\env]{\refine[\tz]{\typ}{ \tz = \cons(\vec{\tz}) \land
	\overrightarrow{\land \form}}[\vec{\tv}/\vec{\tz}]}$. That is,
	$$\ksem{\env}{\kost}{\cons(\vec{\tv})} \in
	\rtsem[{\env\{\vec{\tz}:=\ksem{\env}{\kost}{\vec{N}}\}}]{\refine[\tz]{\typ}{
	\tz = \cons(\vec{\tz}) \land \overrightarrow{\land \form}}}.$$
  Since,
	by the induction hypothesis, we have that
	$\overrightarrow{\form[\ksem{\env}{\kost}{\tv}]}$ holds, we are done.

      \item Rule \ref{rty:case}, we have
	$\rctx \INFr\caseof{\tv}[\cons(\vec{\tx})]{\te_0}[\txtwo]{\te_1}: \rty$ as a
	consequence of
	$\rctx\INFr \tv : \refine{\typ}{\form}$,
	$\rctx, \overrightarrow{\tx : \typ}; \form,\tv = \cons(\vec{\tx})
	\INFr \te_0 : \rty$, and
	$\rctx, \txtwo : \refine{\typ}{\form}; \txtwo = \tv,
	\forall \vec{X}^{\vec{\typ}}.\ \txtwo \neq \cons(\vec{\tx}) \INFr \te_1
	: \rty$, where $\cons :: \vec \typ  \to \typ$.

	We need to prove that $\forall \env \in \rtsem{\rctx}$,
	$$\ksem{\env}{\kost}{\caseof{\tv}[\cons(\vec{\tx})]{\te_0}[\txtwo]{\te_1}}
	\in \rtsem[\env]{\rty}.$$
	That is, 
	$\ksem{\env\{\vec X:=\ksem{\env}{\kost}{\vec{\tv}}\}}{\kost}{\te_0} \in \rtsem[\env]{\rty}$ when $\ksem{\env}{\kost}{\tv} =
	\cons(\ksem{\env}{\kost}{\vec{\tv}})$ and
	$\ksem{\env\{\txtwo:=\ksem{\env}{\kost}{\vec{\tv}}\}}{\kost}{\te_1} \in
	\rtsem[\env]{\rty}$ otherwise; which are direct consequences from the
	induction hypothesis.

      \item Rule \ref{rty:rec}, we have $\rctx,\rctx'\INFr
	\letrec{\tf}{\tx}{\te} : \rrfty$ as a consequence of $\rctx,\tf:
	\rrfty,\rctx'\INFr  \lambda \tx .\te : \rrfty$, where
	$\text{Admis}(\rctx\INFwf  \rrfty)$ and $\tf\notin\mathsf{fv}(\rctx')$.
	We need to prove that
	$\forall \env \in \rtsem{\rctx,\rctx'}$,
	$\ksem{\env}{\kost}{\letrec{\tf}{\tx}{\te}} \in \rtsem[\env]{\rrfty}$.
	By the induction hypothesis, $\forall v\in\rtsem[\env]{\rrfty}$,
	$\ksem{\env\{\tf:=v\}}{\kost}{\lambda \tx .\te} \in
	\rtsem[\env\{\tf:=v\}]{\rrfty}$, and since we have $\text{Admis}(\rctx\INFwf  \rrfty)$, we are done.

      \item Rule \ref{rty:real}, we have $\rctx\INFr \breal{r} :
	\refine{\realty}{\tz = \breal{r}}$.  We need to prove that $\forall
	\env \in \rtsem{\rctx}$, $\ksem{\env}{\kost}{\breal{r}} \in
	\rtsem[\env]{\refine{\realty}{\tz = \breal{r}}}$, or equivalently,
	$\breal{r} \in \{v \in \Rext \mid \rtsem[\env\{\tz = v\}]{\tz =
	\breal{r}}\}$, which is trivial.

      \item Rule \ref{rty:opplus}, we have $\rctx\INFr \te_0 \cadd
	\te_1 : \refine{\kty}{\exists \tz_0^{\realty}, \tz_1^{\kty}.\ \tz =
	\tz_0 \cadd \tz_1 \land \form_0[\tz_0] \land \form_1[\tz_1]}$ as a
	consequence of $\rctx\INFr \te_0 : \refine{\realty}{\form_0}$
	and $\rctx\INFr \te_1 : \refine{\kty}{\form_1}$.

	We need to prove that $\forall \env \in \rtsem{\rctx}$,
	\begin{multline*}\ksem{\env}{\kost}{\te_0 \cadd \te_1} \in\\
	\rtsem[\env]{\refine{\kty}{\exists \tz_0^{\realty}, \tz_1^{\kty}.\ \tz
	= \tz_0 \cadd \tz_1 \land \form_0[\tz_0] \land \form_1[\tz_1]}}.\end{multline*}
	Since, by the induction hypothesis,
	$\form_0[\ksem{\env}{\kost}{\te_0}]$ and
	$\form_1[\ksem{\env}{\kost}{\te_1}]$ hold, we are done.

      \item Rule \ref{rty:bary}, we have
	\begin{multline*}\rctx\INFr \te_0\cbary{\tv} \te_1  :\\ \{ \tz : \kty \mid
	{\exists \tz_0^{\kty}, \tz_1^{\kty}.\ \tz = \tz_0 \cbary{\tv} \tz_1
        \land \form_0[\tz_0 ] \land\form_1[\tz_1 ]\land \form_2[\tv ]}\}\end{multline*}
	as a consequence of
	$\rctx\INFr \te_0 : \refine{\kty}{\form_0}$,
	$\rctx\INFr \te_1 : \refine{\kty}{\form_1}$,
	and $\rctx\INFr \tv : \refine{\qbty}{\form_2}$.

	We need to prove that $\forall \env \in \rtsem{\rctx}$, we have
\begin{multline*}
\ksem{\env}{\kost}{\te_0\cbary{\tv} \te_1} \in\\ \rtsem[\env]{\{\tz: \kty \mid \exists \tz_0^{\kty}, \tz_1^{\kty}.\ \tz = \tz_0 \cbary{\tv} \tz_1 \land \form_0[\tz_0] \land \form_1[\tz_1] \land \form_2[\tv]\}}.
\end{multline*}

	Since, by the induction hypothesis, we have that
	$\form_0[\ksem{\env}{\kost}{\te_0}]$,
	$\form_1[\ksem{\env}{\kost}{\te_1}]$, and
	$\form_2[\ksem{\env}{\kost}{\tv}]$ hold, we are done.

      \item Rule \ref{rty:RTsub}, we have $\rctx\INFr \te : \rty$ as a
	consequence of $\rctx\INFr \te : \rty'$ and $\rctx\INFs \rty' \rsubtype
	\rty$.  We need to prove that $\forall \env \in \rtsem{\rctx}$,
	$\ksem{\env}{\kost}{\te} \in \rtsem[\env]{\rty}$.  By the induction
	hypothesis, $\ksem{\env}{\kost}{\te} \in \rtsem[\env]{\rty'}$ and by
	Lemma~\ref{lem:subtypingSoundness},
	$\rtsem[\env]{\rty'} \subseteq \rtsem[\env]{\rty}$, thus, we are done.

      \item Rule \ref{rty:Gen}, we have $\rctx,\rctx' \INFr \te : \forall \tx :
	\rty'.\ \rty$ as a consequence of $\rctx, \tx : \rty',\rctx'\INFr \te :
	\rty$, with $\tx \notin \mathsf{fv}(\rctx')\cup\mathsf{fv}(\te)$.  We
	need to prove that $\forall \env \in \rtsem{\rctx}$,
	$\ksem{\env}{\kost}{\te} \in \rtsem[\env]{\forall \tx : \rty'.\ \rty}$.
	That is, $\ksem{\env}{\kost}{\te} \in\typesem{\skel{\rty}}$ and
	$\forall v'\in\rtsem[\env]{\rty'}$,
	$\ksem{\env}{\kost}{\te}\in\rtsem[\env\{\tx:=v'\}]{\rty}$.  By the
	induction hypothesis, $\ksem{\env\{\tx:=v'\}}{\kost}{\te} \in
	\rtsem[\env\{\tx:=v'\}]{\rty}$, and since $\tx \notin
	\mathsf{fv}(\te)$, we are done.

      \item Rule \ref{rty:Inst}, we have $\rctx \INFr \te : \rty[\tv/\tx]$ as a consequence of
	  $\rctx \INFr \te : \forall \tx : \rty'.\ \rty$ and $\rctx \INFr \tv : \rty'$.
	  We need to prove that $\forall \env \in \rtsem{\rctx}$, $\ksem{\env}{\kost}{\te} \in \rtsem[\env]{\rty[\tv/\tx]}$, which, using Lemma~\ref{lem:subst}, is the same to prove that $\ksem{\env}{\kost}{\te} \in \rtsem[\env\{\tx:=\ksem{\env}{\kost}{\tv}\}]{\rty}$.
	  By the induction hypothesis, $\ksem{\env}{\kost}{\te} \in \rtsem[\env]{\forall \tx : \rty'.\ \rty}$ and $\ksem{\env}{\kost}{\tv} \in \rtsem[\env]{\rty'}$, thus, we are done.
	  \qedhere
    \end{itemize}
\end{proof}

\end{document}